\documentclass[review]{elsarticle}

\usepackage{lineno,hyperref,amsmath,amsthm,amsfonts}\usepackage{algorithm}
\usepackage{algpseudocode}
\modulolinenumbers[5]
\newtheorem{proposition}{Proposition}

\journal{Journal of Computational Physics}

\begin{document}

\begin{frontmatter}

\title{Efficient Statistically Accurate Algorithms for the Fokker-Planck Equation in Large Dimensions}

\author{Nan Chen\fnref{myfootnote}}
\address{Department of Mathematics and Center for Atmosphere
Ocean Science, Courant Institute of Mathematical Sciences, New York University, New York, NY, USA.}
\fntext[myfootnote]{Corresponding author}
\ead{chennan@cims.nyu.edu}

\author{Andrew J. Majda}
\address{Department of Mathematics and Center for Atmosphere
Ocean Science, Courant Institute of Mathematical Sciences, New York University, New York, NY, USA, \\ Center for Prototype
Climate Modeling, New York University Abu Dhabi, Saadiyat Island, Abu Dhabi, UAE.}
\ead{jonjon@cims.nyu.edu}




\begin{abstract}
Solving the Fokker-Planck equation for high-dimensional complex turbulent dynamical systems is an important and practical issue. However, most traditional methods suffer from the curse of dimensionality and have difficulties in capturing the fat tailed highly intermittent probability density functions (PDFs) of complex systems in turbulence, neuroscience and excitable media. In this article, efficient statistically accurate algorithms are developed for solving both the transient and the equilibrium solutions of Fokker-Planck equations associated with high-dimensional nonlinear turbulent dynamical systems with conditional Gaussian structures. The algorithms involve a hybrid strategy that requires only a small number of ensembles. Here, a conditional Gaussian mixture in a high-dimensional subspace via an extremely efficient parametric method is combined with a judicious non-parametric Gaussian kernel density estimation in the remaining low-dimensional subspace. Particularly, the parametric method provides closed analytical formulae for determining the conditional Gaussian distributions in the high-dimensional subspace and is therefore computationally efficient and accurate. The full non-Gaussian PDF of the system is then given by a Gaussian mixture. Different from the traditional particle methods, each conditional Gaussian distribution here covers a significant portion of the high-dimensional PDF. Therefore a small number of ensembles is sufficient to recover the full PDF, which overcomes the curse of dimensionality. Notably, the mixture distribution has a significant skill in capturing the transient behavior with fat tails of the high-dimensional non-Gaussian PDFs, and this facilitates the algorithms in accurately describing the intermittency and extreme events in complex turbulent systems. It is shown in a stringent set of test problems that the method only requires an order of $O(100)$ ensembles to successfully recover the highly non-Gaussian transient PDFs in up to $6$ dimensions with only small errors.

\end{abstract}


\end{frontmatter}


\section{Introduction}

The Fokker-Planck equation describes the time evolution of the probability density function (PDF) of complex systems with noise \citep{gardiner1985stochastic, risken1989fokker}. Solving the Fokker-Planck equation for both the steady state and transient phases in high dimensions is an important problem in science, engineering, finance, and many other areas. In addition to the large dimensions, strong non-Gaussianity due to the nonlinear coupling and state-dependent noise in the underlying  dynamical systems is another salient feature of the PDFs in many applications, such as geophysical and engineering turbulence, neuroscience and excitable media \cite{majda2016introduction, lindner2004effects}. Examples include the prediction of extreme events \citep{chen2014predicting, cousins2014quantification, ghil2011extreme, palmer2002quantifying} and rare events \citep{mohamad2016probabilistic, cousins2016reduced, thual2016simple}, the uncertainty quantification of the systems with intermittent instability \citep{majda2012lessons, branicki2012quantifying, greco2009statistical} and the characterization of other non-Gaussian events in nature \citep{neelin2010long, huang2001application}. These intermittency and extreme/rare events usually result in strong skewness and fat tails in the non-Gaussian PDFs.

Now let's consider a general nonlinear dynamical system with noise,
\begin{equation}\label{General_SDE}
  d\mathbf{u} = \mathbf{F}(\mathbf{u},t)dt + \boldsymbol\Sigma(\mathbf{u},t)d\mathbf{W},
\end{equation}
with state variables $\mathbf{u}\in \mathbb{R}^N$, noise matrix $\boldsymbol\Sigma\in\mathbb{R}^{N\times K}$ and white noise $\mathbf{W}\in\mathbb{R}^K$. The following partial differential equation (PDE) is the so-called Fokker-Planck equation \citep{gardiner1985stochastic, risken1989fokker} that describes the evolution of the smooth PDF $p(\mathbf{u},t)$  associated with \eqref{General_SDE},
\begin{equation}\label{Fokker_Planck_Equation}
\begin{split}
  \frac{\partial}{\partial t} p(\mathbf{u},t) &= -\nabla_\mathbf{u}\big(\mathbf{F}(\mathbf{u},t)p(\mathbf{u},t)\big) + \frac{1}{2}\nabla_\mathbf{u}\cdot\nabla_\mathbf{u}(\mathbf{Q}(\mathbf{u},t)p(\mathbf{u},t)),\\
  p_t\big|_{t=t_0} &= p_0(\mathbf{u}),
\end{split}
\end{equation}
with $\mathbf{Q} = \boldsymbol\Sigma\boldsymbol\Sigma^T$.

Since there is no general analytical solution for the Fokker-Planck equation \eqref{Fokker_Planck_Equation}, various numerical approaches are developed for solving the evolution of $p(\mathbf{u},t)$. Among these methods, finite element and finite difference are widely used. However, the enormous computational cost makes these PDE solvers impractical for systems with dimension larger than three \citep{pichler2013numerical, kumar2006solution, spencer1993numerical}. Another well-known approach of solving $p(\mathbf{u},t)$ is through the direct Monte Carlo simulation of \eqref{General_SDE}. Unfortunately, the same curse of dimensionality problem appears, where the sample size increases in an exponential rate as the dimension $N$ \citep{robert2004monte, daum2003curse}. In addition, a substantial number of Monte Carlo samples is already required even in the low-dimensional scenarios in order to recover the fat tails of the highly intermittent non-Gaussian PDFs with accuracy \citep{ackerman2010fat}. On the other hand, there are a few methods that work for the approximate solutions of the Fokker-Planck equation with dimension larger than three for some special types of the dynamical systems. For example, asymptotic expansion with truncations can be applied to systems with multiscale structures. The solution of the Fokker-Planck equation associated with the truncated system provides a good approximation for the time evolution of large-scale or slowly varying variables \citep{gardiner1985stochastic, majda1999models, majda2001mathematical, majda2006stochastic}. With extra conditions for both nonlinear and noise terms, splitting methods also provide reasonably good estimations of the PDF for systems with weak non-Gaussianity \citep{er2011methodology, er2012state}. In addition, orthogonal functions and tensor decompositions have been applied to solve the steady state solution of \eqref{Fokker_Planck_Equation} for some class of models \citep{von2000calculation, sun2014numerical, risken1989fokker}.

In this article, efficient statistically accurate algorithms are developed for solving the Fokker-Planck equation associated with high-dimensional nonlinear
turbulent dynamical systems with conditional Gaussian structures \citep{chen2016filtering}. Decomposing $\mathbf{u}$ in \eqref{General_SDE} into two groups of variables $\mathbf{u}=(\mathbf{u}_\mathbf{I},\mathbf{u}_\mathbf{II})$ with $\mathbf{u}_{\mathbf{I}}\in R^{N_{\mathbf{I}}}$ and $\mathbf{u}_{\mathbf{II}}\in R^{N_{\mathbf{II}}}$, the conditional Gaussian systems are characterized by the fact that once a single trajectory of $\mathbf{u}_\mathbf{I}(s\leq t)$ is given, $\mathbf{u}_\mathbf{II}(t)$ conditioned on $\mathbf{u}_\mathbf{I}(s\leq t)$ becomes a Gaussian process. Despite the conditional Gaussianity, the coupled systems remain highly nonlinear and is able to capture strong non-Gaussian features such as skewed or fat-tailed distributions as observed in nature \citep{chen2016filtering}. One of the desirable features of such conditional Gaussian system is that it allows closed analytical formulae for solving the conditional distribution $p(\mathbf{u}_\mathbf{II}(t)|\mathbf{u}_\mathbf{I}(s\leq t))$ based on a Bayesian framework \citep{liptser2001statistics}. Note that most turbulent dynamical systems contain only a small dimension of the observed variables $\mathbf{u}_{\mathbf{I}}$ that represent large scales or surface variables while the dimension of $\mathbf{u}_{\mathbf{II}}$ can be very large. Applications of the conditional Gaussian systems to strongly nonlinear systems include predicting the intermittent time-series of the Madden-Julian oscillation (MJO) and monsoon intraseasonal variabilities \citep{chen2014predicting, chen2015predicting, chen2015predicting2}, filtering the stochastic skeleton model for the MJO \citep{chen2016filtering2}, and recovering the turbulent ocean flows with noisy observations from Lagrangian tracers \citep{chen2014information, chen2015noisy, chen2016model}. Other studies that also fit into the conditional Gaussian framework includes the cheap exactly solvable forecast models in dynamic stochastic superresolution of sparsely observed turbulent systems \citep{branicki2013dynamic, keating2012new}, stochastic superparameterization for geophysical turbulence \citep{majda2014new}, physics constrained nonlinear regression models \citep{majda2012physics, harlim2014ensemble} and blended particle filters for large-dimensional chaotic systems \citep{majda2014blended}.

Different from the traditional particle methods, these efficient statistically accurate algorithms employ only a small number $L$ of ensembles. In fact, all that is required in the algorithms is $L$ independent trajectories of the low-dimensional variables $\mathbf{u}_\mathbf{I}$. Here a hybrid strategy is involved in these efficient statistically accurate algorithms, where a conditional Gaussian mixture with $L$ components in the high-dimensional subspace of $\mathbf{u}_{\mathbf{II}}$ via an extremely efficient parametric method is combined with a judicious non-parametric Gaussian kernel density estimation in the low-dimensional subspace of $\mathbf{u}_{\mathbf{I}}$. Despite the high dimensionality, each component of the conditional Gaussian mixture is computed via the closed analytical formulae and the $L$ components can even be solved in a parallel way due to their independence. Therefore, this parametric method for solving the conditional Gaussian mixture in the high-dimensional subspace is computationally efficient and accurate. Then combining each component of the conditional Gaussian mixture of $\mathbf{u}_{\mathbf{II}}$ with the corresponding Gaussian distribution of $\mathbf{u}_{\mathbf{I}}$ from the Gaussian kernel method results in a Gaussian mixture for the full PDF $p(\mathbf{u}_\mathbf{I}, \mathbf{u}_\mathbf{II})$. One of the compelling features of the algorithms is that each conditional Gaussian distribution is able to cover a significant portion of the high-dimensional PDF $p(\mathbf{u}_\mathbf{II})$. This is the fundamental reason that a small number of ensembles is sufficient in recovering the full PDF, which greatly ameliorates the curse of dimensionality. In particular, the mixture distribution has a significant skill in capturing the fat tails of the high-dimensional non-Gaussian PDFs that are associated with the intermittency and extreme events in the turbulent systems. In addition, the solution resulting from the algorithms converges to the PDE solution of the Fokker-Planck equation with no intrinsic barrier as in those approximate methods. Practically, with $L\sim O(100)$, this hybrid approach is able to recover the joint PDF with $\mbox{Dim}(\mathbf{u}_{\mathbf{I}})\leq3$ and $\mbox{Dim}(\mathbf{u}_{\mathbf{II}})\sim O(10)$. Note that the idea of adopting hybrid methods have also been applied in data assimilation and ensemble prediction in high dimensional turbulent systems with non-Gaussian features \citep{majda2014blended, sapsis2013blending, sapsis2013statistically,  slivinski2015hybrid, hamill2000hybrid}.

The remaining of this article is organized as follows. A general framework of the turbulent dynamical systems with conditional Gaussian structures is shown in Section \ref{Sec:Systems} with concrete examples that involve intermittency and extreme events. Section \ref{Sec:Algorithm} presents the efficient statistically accurate algorithms. Section \ref{Sec:Performance} includes the performance tests of the algorithms for high-dimensional non-Gaussian turbulent dynamical systems. Conclusion and discussions are given in Section \ref{Sec:Conclusion}. The details of an important family of the conditional Gaussian systems with energy-conserving nonlinear interactions that appears in many turbulent applications are shown in the Appendix.

\section{High-Dimensional Conditional Gaussian Models Exhibiting Nonlinear Dynamics with Extreme Events, Intermittency and Other Complex Non-Gaussian Features}\label{Sec:Systems}

The general framework of high dimensional conditional Gaussian models is given as follows \citep{liptser2001statistics, chen2016filtering}:
\begin{subequations}\label{Conditional_Gaussian_System}
\begin{align}
    d\mathbf{u}_{\mathbf{I}} &= [\mathbf{A}_0(t,\mathbf{u}_{\mathbf{I}})+\mathbf{A}_1(t,\mathbf{u}_{\mathbf{I}})\mathbf{u}_{\mathbf{II}}]dt + \boldsymbol{\Sigma}_{\mathbf{I}}(t,\mathbf{u}_{\mathbf{I}})d\mathbf{W}_{\mathbf{I}}(t),\label{Conditional_Gaussian_System1}\\
    d\mathbf{u}_{\mathbf{II}} &= [\mathbf{a}_0(t,\mathbf{u}_{\mathbf{I}})+\mathbf{a}_1(t,\mathbf{u}_{\mathbf{I}})\mathbf{u}_{\mathbf{II}}]dt + \boldsymbol{\Sigma}_{\mathbf{II}}(t,\mathbf{u}_{\mathbf{I}})d\mathbf{W}_{\mathbf{II}}(t), \label{Conditional_Gaussian_System2}
\end{align}
\end{subequations}
where the state variables are written in the form $\mathbf{u} = (\mathbf{u}_{\mathbf{I}}, \mathbf{u}_{\mathbf{II}})$ with both $\mathbf{u}_{\mathbf{I}}\in R^{N_\mathbf{I}}$ and $\mathbf{u}_{\mathbf{II}}\in R^{N_\mathbf{II}}$ being multidimensional variables. In \eqref{Conditional_Gaussian_System}, $\mathbf{A}_0, \mathbf{A}_1, \mathbf{a}_0, \mathbf{a}_1, \boldsymbol{\Sigma}_{\mathbf{I}}$ and $\boldsymbol{\Sigma}_{\mathbf{II}}$ are vectors and matrices that depend only on time $t$ and the state variables $\mathbf{u}_{\mathbf{I}}$, and $\mathbf{W}_{\mathbf{I}}(t)$ and $\mathbf{W}_{\mathbf{II}}(t)$ are independent Wiener processes. The systems in \eqref{Conditional_Gaussian_System} are named as conditional Gaussian systems due to the fact that once $\mathbf{u}_{\mathbf{I}}(s)$ for $s\leq t$ is given, $\mathbf{u}_{\mathbf{II}}(t)$ conditioned on $\mathbf{u}_{\mathbf{I}}(s)$ becomes a Gaussian process with mean $\mathbf{\bar{u}}_{\mathbf{II}}(t)$ and covariance $\mathbf{R}_{\mathbf{II}}(t)$, i.e.,
\begin{equation}\label{CG_PDF}
    p\big(\mathbf{u}_{\mathbf{II}}(t)|\mathbf{u}_{\mathbf{I}}(s\leq t)\big) \sim \mathcal{N}(\mathbf{\bar{u}}_{\mathbf{II}}(t), \mathbf{R}_{\mathbf{II}}(t)).
\end{equation}
Despite the conditional Gaussianity, the coupled system \eqref{Conditional_Gaussian_System} remains highly nonlinear and is able to capture the strong non-Gaussian features as observed in nature \citep{chen2016filtering}. One of the desirable features of the conditional Gaussian system \eqref{Conditional_Gaussian_System} is that the conditional distribution in \eqref{CG_PDF} has the following closed analytical form \citep{liptser2001statistics},
\begin{equation}\label{CG_Result}
\begin{aligned}
        d\mathbf{\bar{u}}_{\mathbf{II}}(t) = & [\mathbf{a}_0(t,\mathbf{u}_{\mathbf{I}})+\mathbf{a}_1(t,\mathbf{u}_{\mathbf{I}})\mathbf{\bar{u}}_{\mathbf{II}}]dt+(\mathbf{R}_{\mathbf{II}}\mathbf{A}^{*}_1(t,\mathbf{u}_{\mathbf{I}}))
        (\boldsymbol{\Sigma}_{\mathbf{I}}\boldsymbol{\Sigma}_{\mathbf{I}}^*)^{-1}(t,\mathbf{u}_{\mathbf{I}})\times\\
        &\qquad\qquad\qquad\qquad
    [d\mathbf{u}_{\mathbf{I}}-(\mathbf{A}_0(t,\mathbf{u}_{\mathbf{I}})+\mathbf{A}_1(t,\mathbf{u}_{\mathbf{I}})\mathbf{\bar{u}}_{\mathbf{II}})dt],\\
    d\mathbf{R}_{\mathbf{II}}(t) = & \left\{\mathbf{a}_1(t,\mathbf{u}_{\mathbf{I}})\mathbf{R}_{\mathbf{II}}+\mathbf{R}_{\mathbf{II}}\mathbf{a}^{*}_1(t,\mathbf{u}_{\mathbf{I}})+(\boldsymbol{\Sigma}_{\mathbf{II}}\boldsymbol{\Sigma}_{\mathbf{II}}^*)(t,\mathbf{u}_{\mathbf{I}})\right.\\
    &\qquad\qquad\left.-(\mathbf{R}_{\mathbf{II}}\mathbf{A}^{*}_1(t,\mathbf{u}_{\mathbf{I}}))
    (\boldsymbol{\Sigma}_{\mathbf{I}}\boldsymbol{\Sigma}_{\mathbf{I}}^*)^{-1}(t,\mathbf{u}_{\mathbf{I}})(\mathbf{R}_{\mathbf{II}}\mathbf{A}^{*}_1(t,\mathbf{u}_{\mathbf{I}}))^*\right\}dt.
\end{aligned}
\end{equation}

In most geophysical and engineering turbulent dynamical systems, the nonlinear terms are quadratic and the total energy in the nonlinear terms is conserved \citep{majda2016introduction, majda1999models, majda2001mathematical, majda2015statistical, harlim2014ensemble, majda2012physics}. The nonlinear interactions in the turbulent dynamical systems allow the energy transfer between different scales that induces  intermittent instabilities. On the other hand, such linear instabilities are mitigated by energy-conserving quadratic nonlinear interactions that transfer energy back to the linearly stable modes where it is dissipated, resulting in a statistical steady state. Note that in the absence of such energy-conserving nonlinear interactions, the nonlinear turbulent systems will necessarily suffer from non-physical finite-time blow up of statistical solutions as well as pathological behavior of the related invariant measure \citep{majda2012fundamental}. The abstract form of such kind of turbulent dynamical systems is as follows:
\begin{equation}\label{EnergyConserveModel}
  d\mathbf{u} = \big[ (\mathbf{L}+\mathbf{D})\mathbf{u} + \mathbf{B}(\mathbf{u},\mathbf{u}) + \mathbf{F}(t) \big] dt + \boldsymbol{\Sigma}(t,\mathbf{u})d\mathbf{W}(t),
\end{equation}
where $\mathbf{L}$ is a skew-symmetric linear operator representing the $\beta$ effect of Earth's curvature and topography while $\mathbf{D}$ is a negative definite symmetric operator representing dissipative processes such as surface drag, radiative damping and viscosity, etc \citep{salmon1998lectures, thompson2006scaling, majda2006nonlinear, vallis2017atmospheric}. The quadratic operator $\mathbf{B}(\mathbf{u},\mathbf{u})$ conserves the energy by itself so that it satisfies the following:
\begin{equation*}
  \mathbf{u}\cdot\mathbf{B}(\mathbf{u},\mathbf{u})  = 0.
\end{equation*}

A rich class of turbulent models with energy-conserving quadratic nonlinear interactions in \eqref{EnergyConserveModel} belong to the conditional Gaussian systems \eqref{Conditional_Gaussian_System}. See \ref{Appendix:A} for details. In the remaining of this section, we provide a few examples of the conditional Gaussian turbulent dynamical systems with energy-conserving quadratic nonlinear interactions and these models will also be used for the performance tests of the efficient statistically accurate algorithms in Section  \ref{Sec:Performance}. The parameters in these test models are listed in Table \ref{Table_Parameter} and the trajectories and the PDF at the equilibrium states are shown in Figures \ref{Models}--\ref{Models2}. Note that although the focus here is on the conditional Gaussian turbulent dynamical systems with energy-conserving quadratic nonlinear interactions, the algorithms to be developed in Section  \ref{Sec:Performance} work for the general conditional Gaussian systems \eqref{Conditional_Gaussian_System}.\medskip

\emph{1. The noisy Lorenz 63 (L-63) Model:}
\begin{subequations}\label{Lorenz_63}
\begin{align}
    dx &= \sigma( y - x )dt + \sigma_xdW_x, \\
    dy &= \big(x ( \rho - z ) - y \big)dt + \sigma_ydW_y, \\
    dz &= (x y - \beta z)dt + \sigma_zdW_z,
\end{align}
\end{subequations}
The noisy version of L-63 model involves the energy-conserving nonlinear interactions through the quadratic terms $-xz$ and $xy$ in the $y$ and $z$ equations. The system \eqref{Lorenz_63} belongs to the conditional Gaussian framework when $\mathbf{u}_{\mathbf{I}} = x, \mathbf{u}_{\mathbf{II}} = (y,z)^T$ or $\mathbf{u}_{\mathbf{I}} = (y,z)^T, \mathbf{u}_{\mathbf{II}} = x$. With the classical choice of the parameters $\rho=28, \sigma = 10, \beta = 8/3$ \cite{lorenz1963deterministic} and a moderate noise level for all the three noise coefficients $\sigma_x = \sigma_y = \sigma_z = 10$, both the chaotic behaviors in the trajectories and a noisy version of the Lorenz attractor with the butterfly profile can be seen in Column (a) of Figure \ref{Models}.\medskip\medskip

\emph{2. A 4D stochastic climate model:}
\begin{subequations}\label{4D_Test_Model}
\begin{align}
    dx_1 &= \Big(  -x_2(L_{12}  + a_1 x_1 + a_2 x_2) + d_1x_1 + F_1 \notag\\
         &\qquad\qquad\qquad\qquad\qquad\qquad + L_{13} y_1  + b_{123} x_2 y_1\Big) dt + \sigma_{x_1}dW_{x_1}, \\
    dx_2 &= \Big(  +x_1(L_{12}  + a_1 x_1 + a_2 x_2) + d_2x_2 + F_2 \notag\\
         &\qquad\qquad\qquad\qquad\qquad\qquad + L_{24} y_2  + b_{213} x_1 y_1\Big) dt + \sigma_{x_2}dW_{x_1}, \\
    dy_1 &= \Big(-L_{13} x_1 + b_{312} x_1 x_2 + F_3 - \frac{\gamma_1}{\epsilon} y_1 \Big) dt + \frac{\sigma_{y_1}}{\sqrt{\epsilon}}dW_{y_1}. \\
    dy_2 &= \Big(-L_{24} x_2 + F_4 -\frac{\gamma_2}{\epsilon}y_2\Big)dt + \frac{\sigma_{y_2}}{\sqrt{\epsilon}}dW_{y_2},
\end{align}
\end{subequations}
where $ b_{123}+  b_{213}+  b_{312} = 0$. This simple stochastic climate model \citep{majda2008applied, majda2005information} features many of the important dynamical properties of comprehensive global circulation models (GCMs) but with many fewer degree of freedom. It contains a quadratic nonlinear part that conserves energy as well as a linear operator. The linear operator includes a skew-symmetric part that mimics the Coriolis effect and topographic Rossby wave propagation, and a negative definite symmetric part that is formally similar to the dissipation such as the surface drag and radiative damping, as discussed in the general form in \eqref{EnergyConserveModel}. The two variables $x_1$ and $x_2$ can be regarded as climate variables while the other two variables $y_1$ and $y_2$ become weather variables that occur in a much faster time scale when $\epsilon$ is small. The coupling in different variables is through both linear and nonlinear terms, where the nonlinear coupling through $b_{ijk}$ produces multiplicative noise. Note that when $\epsilon\to0$, applying an explicit stochastic mode reduction results in a two-dimensional system for the climate variables \citep{majda1999models, majda2001mathematical, majda2006stochastic}.

The 4D stochastic climate model \eqref{4D_Test_Model} is a conditional Gaussian system with $\mathbf{u}_{\mathbf{I}} = (x_1, x_2)^T$ and $\mathbf{u}_{\mathbf{II}} = (y_1, y_2)^T$. Column (b) in Figure \ref{Models} shows the trajectories and the 1D marginal equilibrium PDFs of a regime with moderate $\epsilon$. In this dynamical regime, highly non-Gaussian marginal equilibrium statistics are found in both the climate variable $x_1$ and the weather variable $y_1$, which is due to the intermittency and extreme events as observed in the trajectories.\medskip\medskip

\emph{3. A nonlinear triad system with multiscale features:}
\begin{subequations}\label{Linear3D}
\begin{align}
  du_1 &= (-\gamma_1 u_1 + L_{12} u_2 + L_{13} u_3 +  I u_1 u_2 + F(t))~dt + \sigma_1 dW_1,\label{Linear3D_u1}\\
  du_2 &= (-L_{12} u_1 - \frac{\gamma_2}{\epsilon} u_2 + L_{23} u_3 - I u_1^2) dt + \frac{\sigma_2}{\epsilon^{1/2}} dW_2,\label{Linear3D_u2}\\
  du_3 &= (-L_{13} u_1 - L_{23} u_2 - \frac{\gamma_3}{\epsilon} u_3  )~dt + \frac{\sigma_3}{\epsilon^{1/2}} dW_3,\label{Linear3D_u3}
\end{align}
\end{subequations}
This nonlinear triad system is a  simple prototype nonlinear stochastic model that mimics structural features of low-frequency variability of GCMs with non-Gaussian features \citep{majda2009normal} and it was used to test the skill  for reduced nonlinear stochastic models for fluctuation dissipation theorem  \citep{majda2010low}. The triad model \eqref{Linear3D} involves a quadratic nonlinear interaction between $u_1$ and $u_2$ with energy-conserving property that induces intermittent instability. On the other hand, the coupling between $u_2$ and $u_3$ is linear and is through the skew-symmetric term with coefficient $-L_{23}$, which represents an oscillation structure of $u_2$ and $u_3$. Particularly, when $L_{23}$ is large, fast oscillations become dominant for $u_2$ and $u_3$ while the overall evolution of $u_1$ can still be slow provided that the feedback from $u_2$ and $u_3$ is damped quickly. Such multiscale structure appears in the turbulent ocean flows described for example by shallow water equation, where $u_1$ stands for the geostrophically balanced part while $u_2$ and $u_3$ mimics the fast oscillations due to the gravity waves \citep{chen2016filtering2}. The large-scale forcing $F(t)$ represents the external time-periodic input to the system, such as the seasonal effects or decadal oscillations in a long time scale \citep{vallis2017atmospheric, mantua2002pacific}.  In addition, the scaling factor $\epsilon$ plays the same role as in the 4D stochastic climate model \eqref{4D_Test_Model} that allows a difference in the memory of the three variables. In Figure \ref{Models2}, the trajectories and the PDFs in three different dynamical regimes are shown. These three dynamical regimes with parameters listed in Table \ref{Table_Parameter} have the following features:
\begin{itemize}
  \item Regime I: a weak coupling between the observed variable $u_1$ and the unobserved variables $u_2, u_3$, large intrinsic noises in $u_2$ and $u_3$, and a moderate $\epsilon$.
  \item Regime II: a strong coupling between the observed variable $u_1$ and the unobserved variables $u_2, u_3$, time-periodic forcing, and a small $\epsilon$.
  \item Regime III: same as Regime II plus strong coupling between $u_2$ and $u_3$ with fast oscillations.
\end{itemize}
Intermittency and extreme events are observed in the trajectories of all the three dynamical regimes. See Figure \ref{Models2}. But each dynamical regime has its own unique features. For example, Regime I has intermittent large-amplitude bursts of instability in $u_1$ and slow evolutions of $u_2$ and $u_3$ due to a moderate $\epsilon$. On the other hand, $u_2$ and $u_3$ in Regime II and III occur in a much faster time scale with $\epsilon=0.1$. Although similar nearly Gaussian PDFs of $u_1$ and highly non-Gaussian PDFs of $u_2$ are found in Regime II and III, the non-Gaussian features of $u_3$ are more significant and the trajectory of $u_2$ has more fluctuations in Regime III due to its strong coupling with fast oscillations between $u_2$ and $u_3$.

This nonlinear triad system belongs to the conditional Gaussian system when  $\mathbf{u}_{\mathbf{I}} = u_1, \mathbf{u}_{\mathbf{II}} = (u_2,u_3)^T$.\medskip\medskip

\emph{4. A 6D conceptual dynamical model for turbulence:}
\begin{subequations}\label{TurbulentModel}
\begin{align}
    du &= ( -d_u u + F_u  + \sum_{i}^5\, \gamma_i\, u\, v_i ) dt + \sigma_u dW_u,\label{TurbulentModel_u}\\
    dv_i &= ( -d_{v_i} v_i - \gamma_i\, u^2) dt + \sigma_{v_i} dW_{v_i},\qquad i = 1,\ldots,5. \label{TurbulentModel_v}
\end{align}
\end{subequations}
This 6D conceptual dynamical model for turbulence is motivated from \citep{majda2014conceptual} and is modified such that it fits into conditional Gaussian framework with $\mathbf{u}_{\mathbf{I}} = u, \mathbf{u}_{\mathbf{II}} = (v_1,\ldots,v_5)^T$, where $u$ is the large-scale observed variable and $v_1$ to $v_5$ represents the unobserved variables from medium to small scales. This 6D conceptual dynamical model for turbulence inherits many important features from the dynamics in \citep{majda2014conceptual}. For example, as shown in Column (c) of Figure \ref{Models}, the large-scale observed variable $u$  and the first unobserved variable $v_1$ are both nearly Gaussian while  small-scale variables $v_3, v_4$ and $v_5$ all have significant fat tails, which are a hallmark of intermittency. In addition, the small-scale turbulent flows provide feedback to large scales through the nonlinear coupling with energy-conserving property.

\begin{table}[h]\label{Table_Parameter}
\begin{center}\hspace*{-2cm}
\begin{tabular}{|l|c|c|c|c|c|c|c|c|c|c|c|c|c|c|c|}
  \hline
  \multicolumn{16}{|l|}{\qquad\qquad Noisy L-63 model \eqref{Lorenz_63}}\\\hline
             & $\sigma$   & $\rho$     & $\beta$    & $\sigma_1$ & $\sigma_2$ & $\sigma_3$ &   &   &   &   &   &  &&&\\
             & $10$       & $28$       & $8/3$      & $10$       & $10$       & $10$       &   &   &   &   &   &  &&&\\\hline
  \multicolumn{16}{|l|}{\qquad\qquad 4D stochastic climate model \eqref{4D_Test_Model}}\\\hline
  & $L_{12}$ & $L_{13}$ & $L_{24}$ & $a_1$ & $a_2$ & $d_1$ & $d_2$ & $\epsilon$ & $\sigma_1$ & $\sigma_2$ & $\sigma_3$ & $\sigma_4$ & $b_{123}$ & $b_{213}$ & $F_i$\\
  & $1$      & $0.5$    & $0.5$    & $2$   & $1$   & $-1$  & $-0.4$& $1$        & $0.5$      & $2$        & $0.5$      & $1$        & $1.5$     & $1.5$     & $0$  \\\hline
   \multicolumn{16}{|l|}{\qquad\qquad 3D nonlinear triad system with multiscale features \eqref{Linear3D}}\\\hline
  Regime     & $\gamma_1$ & $\gamma_2$ & $\gamma_3$ & $L_{12}$ & $L_{13}$ & $L_{23}$ & $\sigma_1$  & $\sigma_2$ & $\sigma_3$ & $I$  & $\epsilon$ & \multicolumn{3}{|c|}{F}&\\
         I   & $2$        & $0.2$      & $0.4$      & $0.2$    & $0.1$    & $0$       & $0.5$      & $1.2$      & $0.8$      & $5$  & 1          & \multicolumn{3}{|c|}{$2$}&\\
         II  & $2$        & $0.6$      & $0.4$      & $1$      & $0.5$    & $0$       & $0.5$      & $0.1$      & $0.1$      & $5$  & 0.1        & \multicolumn{3}{|c|}{$2 + 2\sin(2\pi t)$}&\\
         III & $2$        & $0.6$      & $0.4$      & $1$      & $1$      & $10$      & $0.5$      & $0.1$      & $0.1$      & $5$  & 0.1        & \multicolumn{3}{|c|}{$2 + 2\sin(2\pi t)$}&\\
  \hline
  \multicolumn{16}{|l|}{\qquad\qquad 6D conceptual dynamical model for turbulence \eqref{TurbulentModel}}\\\hline
          & $d_u$ & $F$ & $\sigma_u$ & $\gamma_i$ & $\sigma_{v_1}$ & $\sigma_{v_2}$ & $\sigma_{v_3}$ & $\sigma_{v_4}$ & $\sigma_{v_5}$  & $d_{v_1}$ &  $d_{v_2}$ & $d_{v_3}$ & $d_{v_4}$ & $d_{v_5}$ &\\
          & $0.1$ & $0.5$  & $2.0$   &   $0.25$   & $0.5$          &  $0.2$         &   $0.1$       &   $0.1$       &   $0.1$        &  $0.2$    &  $0.5$     &  $1.0$    & $2.0$     & $5.0$    &\\\hline
\end{tabular}
\caption{Parameters in different nonlinear test models.}\label{Table2_Linear}
\end{center}
\end{table}

\begin{figure}[!h]
{\hspace*{-4cm}\includegraphics[width=20cm]{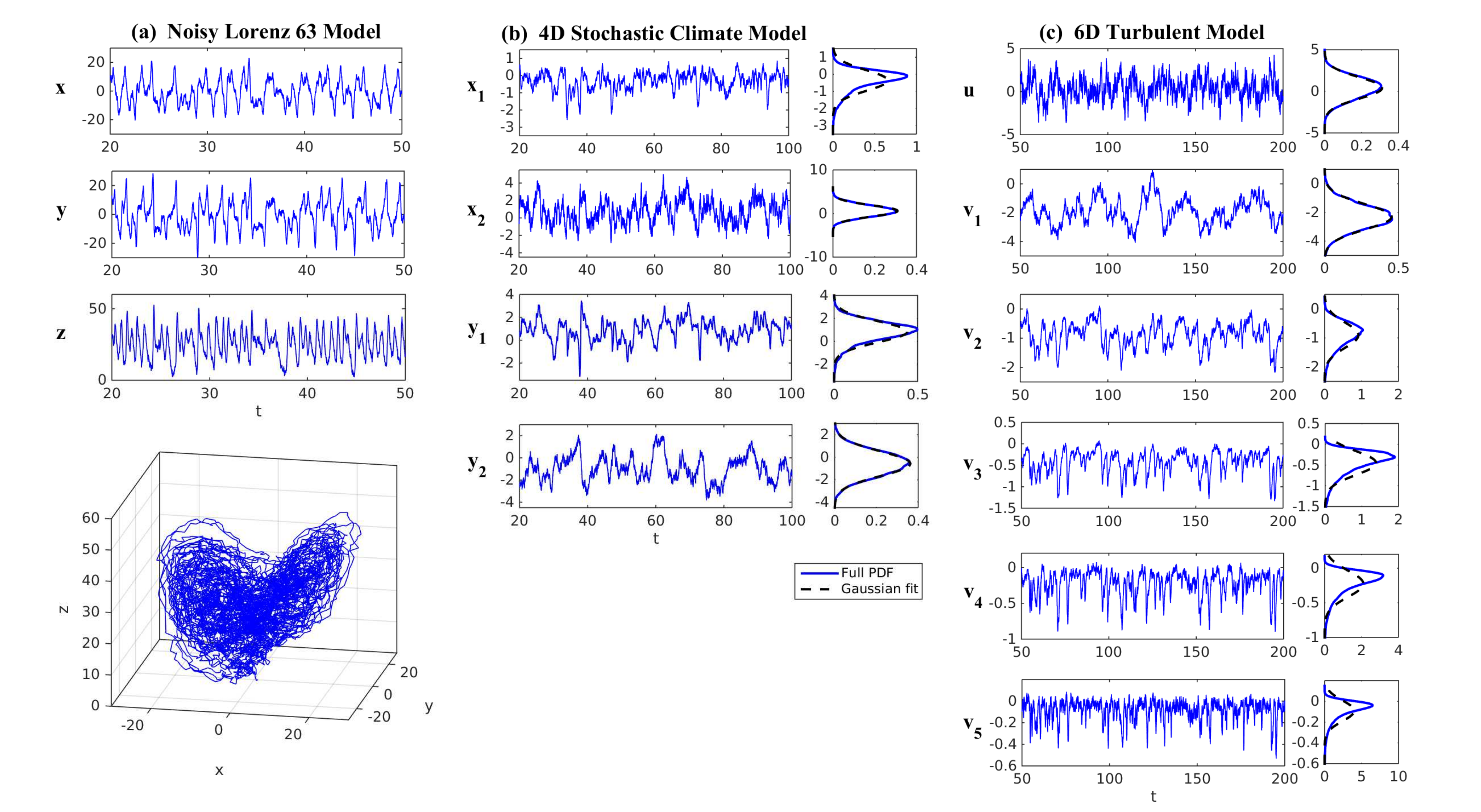}}
\caption{Trajectories and PDFs of different models. (a): Noisy Lorenz 63 model \eqref{Lorenz_63}. (b): 4D stochastic climate model \eqref{4D_Test_Model}. (c): 6D conceptual dynamical model for turbulence \eqref{TurbulentModel}.}\label{Models}
\end{figure}

\begin{figure}[!h]
{\hspace*{-1cm}\includegraphics[width=13cm]{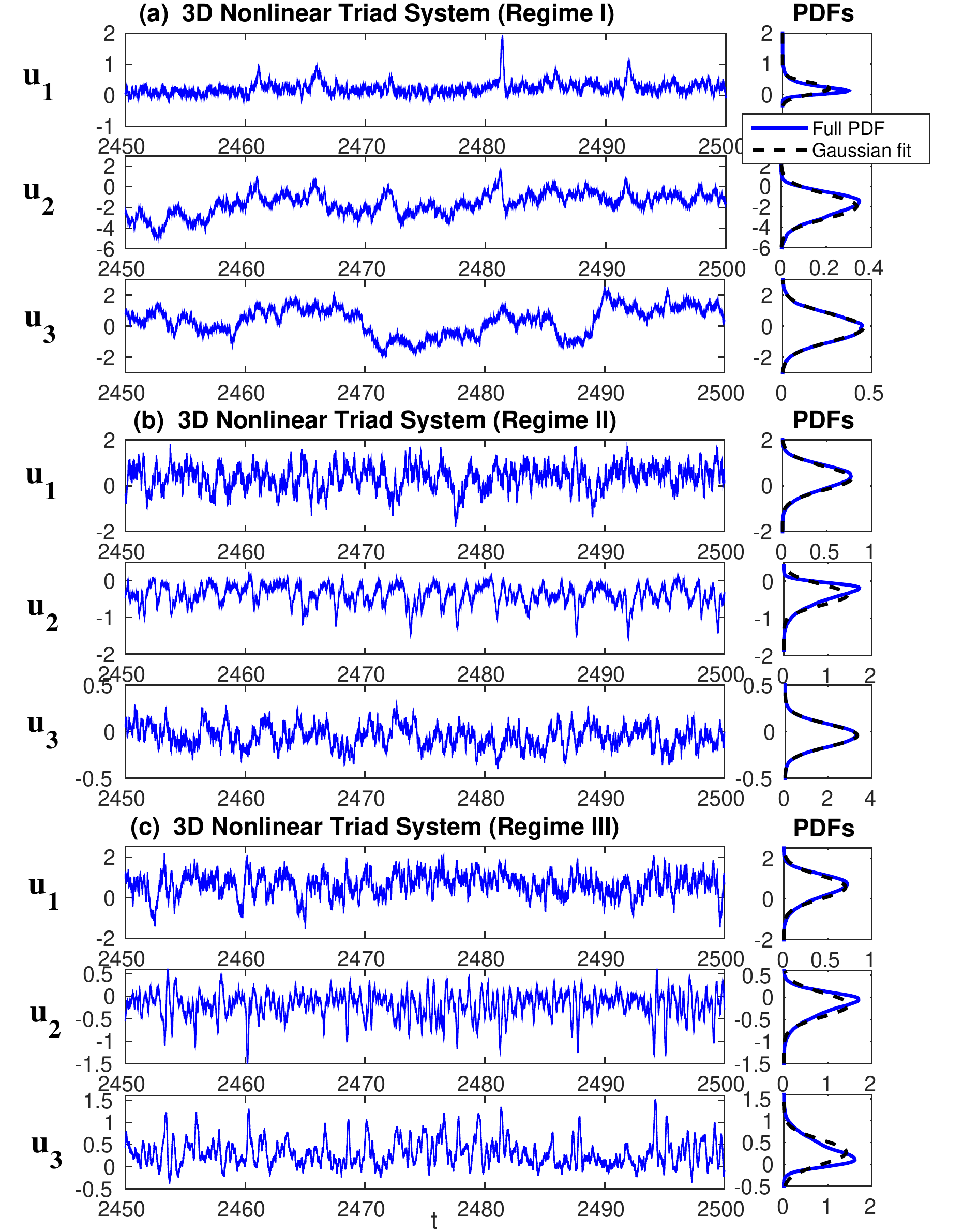}}
\caption{Trajectories and PDFs of the 3D nonlinear triad model \eqref{Linear3D} in three different dynamical regimes with parameters in Table \ref{Table_Parameter}. The PDF in each regime is computed based on one single long trajectory.}\label{Models2}
\end{figure}

\clearpage

\section{Efficient Statistically Accurate Algorithms}\label{Sec:Algorithm}
In this section, we develop efficient statistically accurate algorithms for solving the PDFs associated with the conditional Gaussian turbulent dynamical systems \eqref{Conditional_Gaussian_System}. Recall that $\mathbf{u}_\mathbf{I}\in R^{N_\mathbf{I}}$ and $\mathbf{u}_\mathbf{II}\in R^{N_\mathbf{II}}$. As in most turbulent dynamical systems, we assume the dimension $R^{N_\mathbf{I}}$ of the observed variables is low while that $R^{N_\mathbf{II}}$ of the unobserved variables can be high.

For the data in the algorithms, we generate $L$ independent trajectories in the complex stochastic dynamical systems \eqref{General_SDE}, where $L$ is a small number. In fact,
the only information that is required for these algorithms is $L$ independent trajectories of the observed variables, namely $\mathbf{u}^{1}_\mathbf{I}(s\leq t),\ldots,\mathbf{u}^{L}_\mathbf{I}(s\leq t)$. Practically, since $L$ is small, $\mathbf{u}^{1}_\mathbf{I}(s\leq t),\ldots,\mathbf{u}^{L}_\mathbf{I}(s\leq t)$ can be obtained by running a Monte Carlo simulation for the full system with $L$ samples, which is computationally affordable.
With these $L$ independent trajectories for the observed variables in hand, a hybrid strategy is developed. Here, a parametric method and a non-parametric method are used to deal with the unobserved and observed variables, respectively. Then a Gaussian mixture with block diagonal structure of each mixture component is adopted for solving the full joint PDF $p(\mathbf{u}_{\mathbf{I}}, \mathbf{u}_{\mathbf{II}})$. In the theoretical discussions below, the limit of $L$ going to infinity is taken for the purpose of  mathematical rigor. In the performance tests of the algorithms in Section \ref{Sec:Performance}, $L$ will always be order of $O(100)$. Detailed justifications of adopting such a small $L$ will also be included there.

\subsection{Parametric method for $p(\mathbf{u}_\mathbf{II})$}
First, a parametric method based on the closed form of the conditional Gaussian posterior statistics \eqref{CG_Result} is adopted for solving the PDF of the unobserved variables.
\begin{proposition}\label{Prop:Unobserved}
The marginal distribution of the unobserved variables $\mathbf{u}_{\mathbf{II}}$ at time $t$ is given by the average of the $L$ conditional Gaussian distributions,
\begin{equation*}
   p(\mathbf{u}_{\mathbf{II}}(t)) = \lim_{L\to\infty}\frac{1}{L}\sum_{i=1}^Lp(\mathbf{u}_{\mathbf{II}}(t)|\mathbf{u}^i_{\mathbf{I}}(s\leq t)).
\end{equation*}
\end{proposition}
\begin{proof}
The marginal distribution of $\mathbf{u}_{\mathbf{I}}$ at any fixed time $s$ has the following form
\begin{equation}\label{CG_approximation}
    p(\mathbf{u}_{\mathbf{I}}(s))= \lim_{L\to\infty} \frac{1}{L}\sum_{i=1}^L\delta\left(\mathbf{u}_{\mathbf{I}}(s)-\mathbf{u}^i_{\mathbf{I}}(s)\right).
\end{equation}
Therefore, the  distribution of $\mathbf{u}_{\mathbf{I}}(s\leq t)$ is also equal to the average of the $L$ independent trajectories
\begin{equation}\label{CG_approximation0}
    p(\mathbf{u}_{\mathbf{I}}(s\leq t))= \lim_{L\to\infty} \frac{1}{L}\sum_{i=1}^L\delta\left(\mathbf{u}_{\mathbf{I}}(s\leq t)-\mathbf{u}^i_{\mathbf{I}}(s\leq t)\right),
\end{equation}
According to the fundamental relationship between joint, marginal and conditional distributions, the marginal distribution of $\mathbf{u}_{\mathbf{II}}$ at time $t$ is given by
\begin{equation}\label{CG_approximation2}
\begin{split}
  p(\mathbf{u}_{\mathbf{II}}(t)) &= \int p\Big(\mathbf{u}_{\mathbf{I}}(s\leq t),\mathbf{u}_{\mathbf{II}}(t)\Big)\, d \mathbf{u}_{\mathbf{I}}(s\leq t)\\
   &= \int p\Big(\mathbf{u}_{\mathbf{I}}(s\leq t)\Big)\,p\Big(\mathbf{u}_{\mathbf{II}}(t)|\mathbf{u}_{\mathbf{I}}(s\leq t)\Big)\, d\mathbf{u}_{\mathbf{I}}(s\leq t).
\end{split}
\end{equation}
Inserting \eqref{CG_approximation0} into \eqref{CG_approximation2} yields
\begin{equation}\label{Gaussian_Ensemble}
    p(\mathbf{u}_{\mathbf{II}}(t)) =\lim_{L\to\infty}\frac{1}{L}\sum_{i=1}^Lp\Big(\mathbf{u}_{\mathbf{II}}(t)|\mathbf{u}^i_{\mathbf{I}}(s\leq t)\Big).
\end{equation}
\end{proof}

In \eqref{Gaussian_Ensemble}, given each observational trajectory $\mathbf{u}^i_{\mathbf{I}}(s\leq t)$, the corresponding conditional Gaussian distribution
  \begin{equation}\label{U_II_Mixture}
    p\big(\mathbf{u}_{\mathbf{II}}(t)|\mathbf{u}^i_{\mathbf{I}}(s\leq t)\big) :=p_i\big(\mathbf{u}_{\mathbf{II}}(t)\big)\sim \mathcal{N}(\mathbf{\bar{u}}^i_{\mathbf{II}}(t), \mathbf{R}^i_{\mathbf{II}}(t))
\end{equation}
is solved by the closed analytical formulae \eqref{CG_Result}. In addition, since the $L$ trajectories are independent with each other, these conditional distributions can be computed in a parallel way. Therefore, the algorithm for solving the marginal distribution $p(\mathbf{u}_{\mathbf{I}})$ is computationally efficient. Notably, as $L\to\infty$, \eqref{Gaussian_Ensemble} implies that this algorithm is consistent with solving the Fokker-Planck equation for the marginal PDF $p(\mathbf{u}_{\mathbf{II}})$.

\subsection{Non-parametric method for $p(\mathbf{u}_\mathbf{I})$}
Next, a judicious non-parametric kernel density estimation method is used for solving the PDF of the observed variables.
\begin{proposition}\label{Prop:Observed}
The PDF of $\mathbf{u}_{\mathbf{I}}$ at time $t$ is approximated by a Gaussian kernel density estimation
\begin{equation}\label{Kernel_Method0}
   p\big(\mathbf{u}_{\mathbf{I}}(t)\big)= \lim_{L\to\infty}\frac{1}{L}\sum_{i=1}^L K_\mathbf{H}\Big(\mathbf{u}_{\mathbf{I}}(t)-\mathbf{u}^i_{\mathbf{I}}(t)\Big),
\end{equation}
where $\mathbf{H}$ is the bandwidth, and $K(\cdot)$ is a Gaussian kernel
\begin{equation}\label{U_I_Mixture}
  K_\mathbf{H}\Big(\mathbf{u}_{\mathbf{I}}(t)-\mathbf{u}^i_{\mathbf{I}}(t)\Big):=p_i\big(\mathbf{u}_{\mathbf{I}}(t)\big)\sim \mathcal{N}\Big(\mathbf{u}^i_{\mathbf{I}}(t), \mathbf{H}(t)\Big).
\end{equation}
\end{proposition}
In the limit $L\to \infty$, the kernel density method is simply the Monte Carlo simulation, where the bandwidth shrinks to zero. Therefore, the solution of the kernel density estimation \eqref{Kernel_Method0} is consistent with that of solving the Fokker-Planck equation for the marginal PDF $p(\mathbf{u}_{\mathbf{I}})$.

The kernel density estimation algorithm here involves a ``solve-the-equation plug-in'' approach for optimizing the bandwidth, the idea of which was originally proposed in  \citep{botev2010kernel}. A brief summary of the kernel estimation is shown in \ref{Appendix:B}. Note that the PDFs associated with turbulent dynamical systems are usually highly non-Gaussian. This indicates the failure of the simplest rule-of-thumb bandwidth estimator, which assumes the underlying density being estimated is Gaussian. The solve-the-equation approach does not impose any requirement for the underlying PDF. Therefore, it works for the non-Gaussian cases and the computational cost comes from numerically solving a scalar high order algebraic equation for the optimal bandwidth in order to minimize the asymptotic  mean integrated squared error (AMISE) in the estimator. Finally, kernel density estimations work only for a low dimension space, which is the assumption of $N_{\mathbf{I}}$ of the systems here.

\subsection{Hybrid algorithm for the joint PDF $p(\mathbf{u}_\mathbf{I}, \mathbf{u}_\mathbf{II})$}
With the algorithms for the marginal PDFs of both $\mathbf{u}_{\mathbf{I}}$ and $\mathbf{u}_{\mathbf{II}}$ in hand, a hybrid method is developed to solve the joint PDF.
\begin{proposition}\label{Prop:Both}
The joint PDF of $\mathbf{u}_{\mathbf{I}}$ and $\mathbf{u}_{\mathbf{II}}$ at time $t$ is solved using a Gaussian mixture,
\begin{equation}\label{Joint0}
    p(\mathbf{u}_{\mathbf{I}}(t),\mathbf{u}_{\mathbf{II}}(t)) = \lim_{L\to\infty} \frac{1}{L}\sum_{i=1}^L \Big(K_\mathbf{H}(\mathbf{u}_{\mathbf{I}}(t)-\mathbf{u}_{\mathbf{I}}^i(t))\cdot p(\mathbf{u}_{\mathbf{II}}(t)|\mathbf{u}^i_{\mathbf{I}}(s\leq t))\Big),
\end{equation}
where the two terms in the bracket on the right hand side are both Gaussian and are given by  \eqref{U_II_Mixture} and \eqref{U_I_Mixture}, respectively.
\end{proposition}
\begin{proof}
First, the joint distribution of $\mathbf{u}_{\mathbf{I}}$ and $\mathbf{u}_{\mathbf{II}}$ at time $t$ can be written as
\begin{equation}\label{Joint_derivation1}
  p\Big(\mathbf{u}_{\mathbf{I}}(t),\mathbf{u}_{\mathbf{II}}(t)\Big) = \int p\Big(\mathbf{u}_{\mathbf{II}}(t), \mathbf{u}_{\mathbf{I}}(t)\,|\,\mathbf{u}_{\mathbf{I}}(s\leq t)\Big)\, p\Big(\mathbf{u}_{\mathbf{I}}(s\leq t)\Big) \,d\mathbf{u}_{\mathbf{I}}(s\leq t)
\end{equation}
Here, according to the basic probability relationship $p(x,y|z) = p(x|z)\,p(x|y,z)$, we have the following
\begin{equation}\label{Joint_derivation2}
   p\Big(\mathbf{u}_{\mathbf{II}}(t), \mathbf{u}_{\mathbf{I}}(t)\,|\,\mathbf{u}_{\mathbf{I}}(s\leq t)\Big) =  p\Big(\mathbf{u}_{\mathbf{II}}(t) \,|\,\mathbf{u}_{\mathbf{I}}(s\leq t)\Big)\, p\Big(\mathbf{u}_{\mathbf{I}}(t)\,|\,\mathbf{u}_{\mathbf{I}}(s\leq t)\Big).
\end{equation}
The second term on the right hand side of \eqref{Joint_derivation2} is actually a delta function peaking at the conditioned value of $\mathbf{u}_\mathbf{I}$ at time $t$. In fact, if we replace the condition inside the PDF $\mathbf{u}_{\mathbf{I}}(s\leq t)$ by $\mathbf{u}^i_{\mathbf{I}}(s\leq t)$, we have
\begin{equation}\label{Joint_derivation3}
  p\Big(\mathbf{u}_{\mathbf{I}}(t)\,|\,\mathbf{u}^i_{\mathbf{I}}(s\leq t)\Big) = \delta(\mathbf{u}_{\mathbf{I}}(t)-\mathbf{u}^i_{\mathbf{I}}(t))
\end{equation}
In addition, according to \eqref{CG_approximation0}
\begin{equation}\label{Joint_derivation4}
    p(\mathbf{u}_{\mathbf{I}}(s\leq t))= \lim_{L\to\infty} \frac{1}{L}\sum_{i=1}^L\delta\left(\mathbf{u}_{\mathbf{I}}(s\leq t)-\mathbf{u}^i_{\mathbf{I}}(s\leq t)\right).
\end{equation}
Therefore, inserting \eqref{Joint_derivation2}--\eqref{Joint_derivation4} into \eqref{Joint_derivation1} yields
\begin{equation}\label{Joint_derivation5}
\begin{split}
  p\Big(\mathbf{u}_{\mathbf{I}}(t),\mathbf{u}_{\mathbf{II}}(t)\Big) &= \int p\Big(\mathbf{u}_{\mathbf{II}}(t), \mathbf{u}_{\mathbf{I}}(t)\,|\,\mathbf{u}_{\mathbf{I}}(s\leq t)\Big)\, p\Big(\mathbf{u}_{\mathbf{I}}(s\leq t)\Big) \,d\mathbf{u}_{\mathbf{I}}(s\leq t)\\
  &=\lim_{L\to\infty} \frac{1}{L}\sum_{i=1}^L\delta\Big(\mathbf{u}_{\mathbf{I}}(t)-\mathbf{u}^i_{\mathbf{I}}(t)\Big) \,p\Big(\mathbf{u}_{\mathbf{II}}(t) \,|\,\mathbf{u}^i_{\mathbf{I}}(s\leq t)\Big)
\end{split}
\end{equation}
Next, we make use of the kernel approximation $K_\mathbf{H}(\mathbf{u}_{\mathbf{I}}(t)-\mathbf{u}_{\mathbf{I}}^i(t))$ for $\delta\Big(\mathbf{u}_{\mathbf{I}}(t)-\mathbf{u}^i_{\mathbf{I}}(t)\Big)$. Note that in the limit $L\to\infty$ the bandwidth goes to zero and the kernel approximation converges to $\delta\Big(\mathbf{u}_{\mathbf{I}}(t)-\mathbf{u}^i_{\mathbf{I}}(t)\Big)$, which leads to \eqref{Joint_derivation5} that is consistent with solving the Fokker-Planck equation for the joint PDF.

\end{proof}

Since for each $i$ both $K_\mathbf{H}(\mathbf{u}_{\mathbf{I}}(t)-\mathbf{u}_{\mathbf{I}}^i(t))$ and $p(\mathbf{u}_{\mathbf{II}}(t)|\mathbf{u}^i_{\mathbf{I}}(s\leq t))$ are Gaussian distributions of $\mathbf{u}_\mathbf{I}(t)$ and $\mathbf{u}_\mathbf{II}(t)$, respectively, their combination is also a Gaussian distribution with mean and covariance given as follows:
\begin{equation}\label{MeanCov}
  \mbox{mean} = \left(
                  \begin{array}{c}
                    \mathbf{u}_{\mathbf{I}}^i(t)   \\
                    \mathbf{\bar{u}}^i_{\mathbf{II}}(t) \\
                  \end{array}
                \right).\qquad\mbox{and}\qquad
  \mbox{cov} = \left(
                  \begin{array}{cc}
                    \mathbf{H}(t) &   \\
                      & \mathbf{R}^i_{\mathbf{II}}(t) \\
                  \end{array}
                \right).
\end{equation}
Therefore, the joint distribution $p(\mathbf{u}_{\mathbf{I}}(t),\mathbf{u}_{\mathbf{II}}(t)) $ is a Gaussian mixture. Note that the covariance of each Gaussian component in \eqref{MeanCov} is a block diagonal matrix, which contains no explicit cross-covariance between $\mathbf{u}_\mathbf{I}(t)$ and $\mathbf{u}_\mathbf{II}(t)$. Nevertheless, this does not mean that the cross-correlation between  $\mathbf{u}_\mathbf{I}(t)$ and $\mathbf{u}_\mathbf{II}(t)$ is ignored in this algorithm. Each Gaussian component of $\mathbf{u}_\mathbf{II}(t)$ is solved conditioned on one historical trajectory of $\mathbf{u}_\mathbf{I}(t\leq s)$ and the information of end point of the same observational trajectory $\mathbf{u}_\mathbf{I}(t)$ is used to form the matrix in \eqref{MeanCov}. In fact, it is easy to show that the Gaussian mixture with the block diagonal covariance in \eqref{MeanCov} will converge to the same true PDF as that with a full covariance matrix by making use of the property that the bandwidth $|\mathbf{H}|\to0$ as $L\to\infty$. Details are shown in  \ref{Appendix:C}.

The theoretical justification of the above propositions is shown in the limit of $L\to\infty$. In practice, as long as the dimension $N_{\mathbf{I}}$ of the observed variables $\mathbf{u}_{\mathbf{I}}$ is low, a small number $L$ of the mixture components is sufficient to recover highly non-Gaussian joint PDFs with high accuracy. This is because each conditional Gaussian distribution covers a large portion of $p(\mathbf{u}_{\mathbf{II}}(t))$, which greatly reduces the number of ensembles and ameliorates the curse of dimensionality. It allows the algorithms to be applied to high dimensional systems with $N_\mathbf{II}\gg1$. More detailed discussions will be included in Section \ref{Sec:Performance}, for example Figure \ref{L63_MeanVar} and \ref{3D_MeanVar}. Note that the conditional distributions are obtained via the closed analytical formulae \eqref{CG_Result}, which are  computationally efficient as well.

What remains is to choose the initial values of each conditional Gaussian distribution. Assume that the initial joint distribution is completely known. Given the number $L$, samples $\mathbf{u}^i_{\mathbf{II}}(0), i=1,\ldots,L$ are drawn from $p(\mathbf{u}_{\mathbf{II}}(0))$. In many practical issues, the initial state is deterministic or contains only a small uncertainty. Therefore, each initial conditional Gaussian distribution in \eqref{CG_Result} can be set as $\mathcal{N}(\mathbf{u}^i_\mathbf{II}(0),\epsilon)$, where $\epsilon$ represents a small initial covariance. Nevertheless, in some applications, the initial distribution $p(\mathbf{u}_{\mathbf{II}}(0))$ may contain a large uncertainty. Thus, the following method is incorporated into the algorithm to form the initial conditional Gaussian distribution of each mixture component in \eqref{CG_Result}. Here, the sampled point $\mathbf{u}^i_\mathbf{II}(0)$ is again adopted as the conditional mean of each mixture component and the conditional covariance is computed by the kernel density estimation. Yet, instead of using a direct kernel density estimation for this high dimensional PDF $p(\mathbf{u}_\mathbf{II}(0))$ which is impractical, a diagonal initial covariance matrix is used here, where the variance of each dimension of $\mathbf{u}_{\mathbf{II}}$ is calculated by a 1D kernel density estimation. Since the cross-covariance is already partially reflected by the distribution of the sample points, the simplification with a diagonal covariance is a reasonable choice. In addition, it is easy to show from \eqref{CG_Result} that the conditional covariance converges exponentially fast and therefore the initial error in the conditional covariance will vanish very quickly \citep{chen2017rigorous}. Performance tests in Section \ref{Sec:Performance} will show the recovered PDFs at long, moderate and short transient phases using the strategy discussed above.

Finally, for the convenience of the readers, a pseudo code of the efficient statistically accurate algorithm developed in Proposition \eqref{Prop:Unobserved}--\eqref{Prop:Both} is provided.

\begin{algorithm}
  \caption{}\label{euclid}
  \begin{algorithmic}[0]
    \Procedure{Solving $p(\mathbf{u}_\mathbf{I},\mathbf{u}_\mathbf{II})$ at time $t$}{}
      \State 1. Set the number of samples $L$
      \State 2. Initialization:
        \State\qquad Sample $\big(\mathbf{u}^i_\mathbf{I}(0), \mathbf{u}^i_\mathbf{II}(0)\big)$ for $i =1,\ldots L$
        \State\qquad Set the initial distribution  $(\bar{\mathbf{u}}^i_\mathbf{II}(0), \boldsymbol{R}^i_\mathbf{II}(0))$ for $i =1,\ldots L$
      \State 3. Run Monte Carlo simulation for the full system \eqref{Conditional_Gaussian_System} with $L$ particles:
        \State\qquad Collect the trajectories of the observed variables $\mathbf{u}^i_\mathbf{I}(s\leq t), i=1\ldots,L$
      \State 4. Solve the $L$ conditional Gaussian distributions for $\mathbf{u}_{\mathbf{II}}$ at time $t$
      \For{$i=1,\ldots,L$ }
        \State Run the closed form \eqref{CG_Result} to reach  $p_i(\mathbf{u}_\mathbf{II})\sim\mathcal{N}(\bar{\mathbf{u}}^i_\mathbf{II}, \boldsymbol{R}^i_\mathbf{II})$ at time $t$
        \State [See Proposition \ref{Prop:Unobserved} and Equation \eqref{U_II_Mixture}]
      \EndFor
      \State 5. Solve the $L$ Gaussian distributions for $\mathbf{u}_\mathbf{I}$ at time $t$
      \State \qquad Run kernel method to reach $p_i(\mathbf{u}_\mathbf{I})= K_\mathbf{H}(\mathbf{u}_{\mathbf{I}}(t)-\mathbf{u}_{\mathbf{I}}^i(t))\sim\mathcal{N}({\mathbf{u}}^i_\mathbf{I}, \mathbf{H})$
      \State\qquad [See Proposition \ref{Prop:Observed} and Equation \eqref{U_I_Mixture}]
      \State 6. Form the joint distribution of each Gaussian mixture component
      \For{$i=1,\ldots,L$ }
        \State Use each $p_i(\mathbf{u}_\mathbf{II})$ in 4 and $p_i(\mathbf{u}_\mathbf{I})$ in 5 to form the joint distribution
        \State with mean and covariance given in \eqref{MeanCov}
      \EndFor
      \State 7. Combine the $L$ joint Gaussian distribution to form $p(\mathbf{u}_{\mathbf{I}},\mathbf{u}_{\mathbf{II}})$ at $t$
    \EndProcedure
    \State ${}^*$ To continue solving the PDF $p(\mathbf{u}_\mathbf{I},\mathbf{u}_\mathbf{II})$ at any time $t^\prime>t$, repeat Step 3 to 7 from $t$ to $t^\prime$ with initializations given by the results at time $t$.
    \State ${}^{**}$ The for loop in Step 4 and 6 can be implemented in a parallel way since the manipulations of these $L$ components are independent with each other.
  \end{algorithmic}
\end{algorithm}

\clearpage

\section{Performance tests with highly non-Gaussian features}\label{Sec:Performance}

The performance tests of the efficient statistically accurate algorithms developed in Section \ref{Sec:Algorithm} are illustrated in this section, where the four test models were described in detail in Section \ref{Sec:Systems}.

The natural way to quantify the error in the recovered PDF related to the truth is through an information measure, namely the relative entropy (or Kullback-Leibler divergence) \citep{majda2010quantifying, majda2006nonlinear, majda2011link, branicki2014non, kullback1951information}. The relative entropy is defined as
\begin{equation}\label{Relative_Entropy}
  \mathcal{P}(p(\mathbf{u}),p^M(\mathbf{u})) = \int p(\mathbf{u})\ln\frac{p(\mathbf{u})}{p^M(\mathbf{u})}d\mathbf{u},
\end{equation}
where $p(\mathbf{u})$ is the true PDF and $p^M(\mathbf{u})$ is the recovered one from the efficient statistically accurate algorithms. This asymmetric functional on probability densities $\mathcal{P}(p,p^M)\geq0$ measures lack of information in $p^M$ compared with $p$ and has many attractive features. First, $\mathcal{P}(p,p^M)\geq0$  with equality if and only if $p=p^M$. Secondly, $\mathcal{P}(p,p^M)$ is invariant under general nonlinear changes of variables. Notably, the relative entropy is a good indicator of quantifying the difference in the tails of the two PDFs, which is particularly crucial in the turbulent dynamical systems with intermittency and extreme events. On the other hand, the traditional ways of quantifying the errors, such as the relative error $\|p-p^M\|/\|p\|$, usually underestimate the lack of information in the PDF tails.

In the following performance tests, the number of the observational trajectories adopted in the efficient statistically accurate algorithm is $L=500$ unless stated explicitly otherwise. The full joint PDFs are recovered in all the tests using the algorithm but only 1D and 2D marginal PDFs are shown for the purpose of illustration. The true PDFs in the following tests are formed by Monte Carlo simulations with a huge number of particles $L_{MC} = 150,000$ in order to capture the fat tails in the non-Gaussian distributions.

\subsection{The noisy L-63 model (Equation \eqref{Lorenz_63})}.

Here, the 1D variable $x$ is treated as the observed variable $\mathbf{u}_{\mathbf{I}}$ and the 2D variables $(y,z)^T$ are the unobserved ones $\mathbf{u}_{\mathbf{II}}$. The initial distribution of this test is a multivariate Gaussian distribution with zero mean and unit variance in each direction. The time evolutions of the 1D marginal statistics are shown in Panel (a) of Figure \ref{L63_all}, where the maximum marginal variance and the minimum marginal kurtosis (kurtosis$<2$) of all the three variables occur at a transient phase $t=0.33$. Panel (b) compares the recovered 1D and 2D PDFs with the truth at this transition phase. With only $L=500$, the recovered 1D marginal PDFs succeed in capturing the bimodal characteristics of all the three marginal distributions and the recovered 2D PDFs almost perfectly match the truth which involve highly non-Gaussian features. In Panel (c), the truth and the recovered PDFs are compared at the statistical equilibrium phase $t=1.5$. In addition to the significant skill in recovering the nearly Gaussian 1D marginal PDFs, the algorithm provides an accurate estimation of the non-Gaussian 2D PDF $p(z,x)$ as well.

The skill in the recovered PDFs as a function of $L$ is reported in Figure \ref{L63_ModelError}. Panel (a) illustrates the recovered 2D PDFs with different $L$. The recovered PDFs are already qualitatively similar to the truth with only $L=50$. When $L$ is increased to $100$, the error in the recovered PDFs becomes insignificant. Panel (b) shows the lack of information in the recovered 1D and 2D PDFs related to the truth via the relative entropy \eqref{Relative_Entropy}. The lack of information is an exponential decaying function of $L$. With $L=100$, the lack of information in all the recovered 1D PDFs is below $0.05$ and that in all the recovered 2D ones is below $0.2$.

Figure \ref{L63_MeanVar} shows the posterior mean and posterior variance of $y$ and $z$ associated with each of the $L=500$ mixture components for the tests in Panels (b) and (c) of Figure \ref{L63_all}. The true PDFs and the $L=500$ Monte Carlo sample points of $y$ and $z$ generated in Step 3 of the algorithm are also shown for comparison. Note that the Monte Carlo sample points of $y$ and $z$ are actually never used in the algorithm. They are simply the byproducts of generating the sample trajectories of $x$. Figure \ref{L63_MeanVar} conveys the following messages. First, although $L=500$ sample points from Monte Carlo simulations are able to indicate a rough profile of the 2D PDF, they are still too sparse to provide an accurate estimation. Secondly, there is a compelling difference between the locations of the posterior mean and the Monte Carlo sample points, especially at the highly non-Gaussian transient phase $t=0.33$. In fact, the locations of the posterior mean are solved in an optimal way based on the Bayesian inference and therefore they contain extra information beyond the randomly scattered Monte Carlo sample points. In addition, the posterior variance of different components has distinct values and is significantly larger than zero, which implies that each component is able to cover a large portion of the PDF. These optimized Gaussian distributions with a large covariance guarantees that a small number $L$ is sufficient to recover the PDF with high accuracy. Such a large covariance in each mixture component is particularly important for dealing with the turbulent systems when the dimension $N_{\mathbf{II}}$ is large. Finally, unlike the kernel methods, the posterior variance of different components is independent with each other and does not depend on  $L$ either.  All these properties in the posterior distribution provide evidences that the algorithms developed here are able to greatly ameliorate the curse of dimensionality and work for systems with $N_{\mathbf{II}}\gg1$.

\begin{figure}[!h]
{\hspace*{-3cm}\includegraphics[width=18cm]{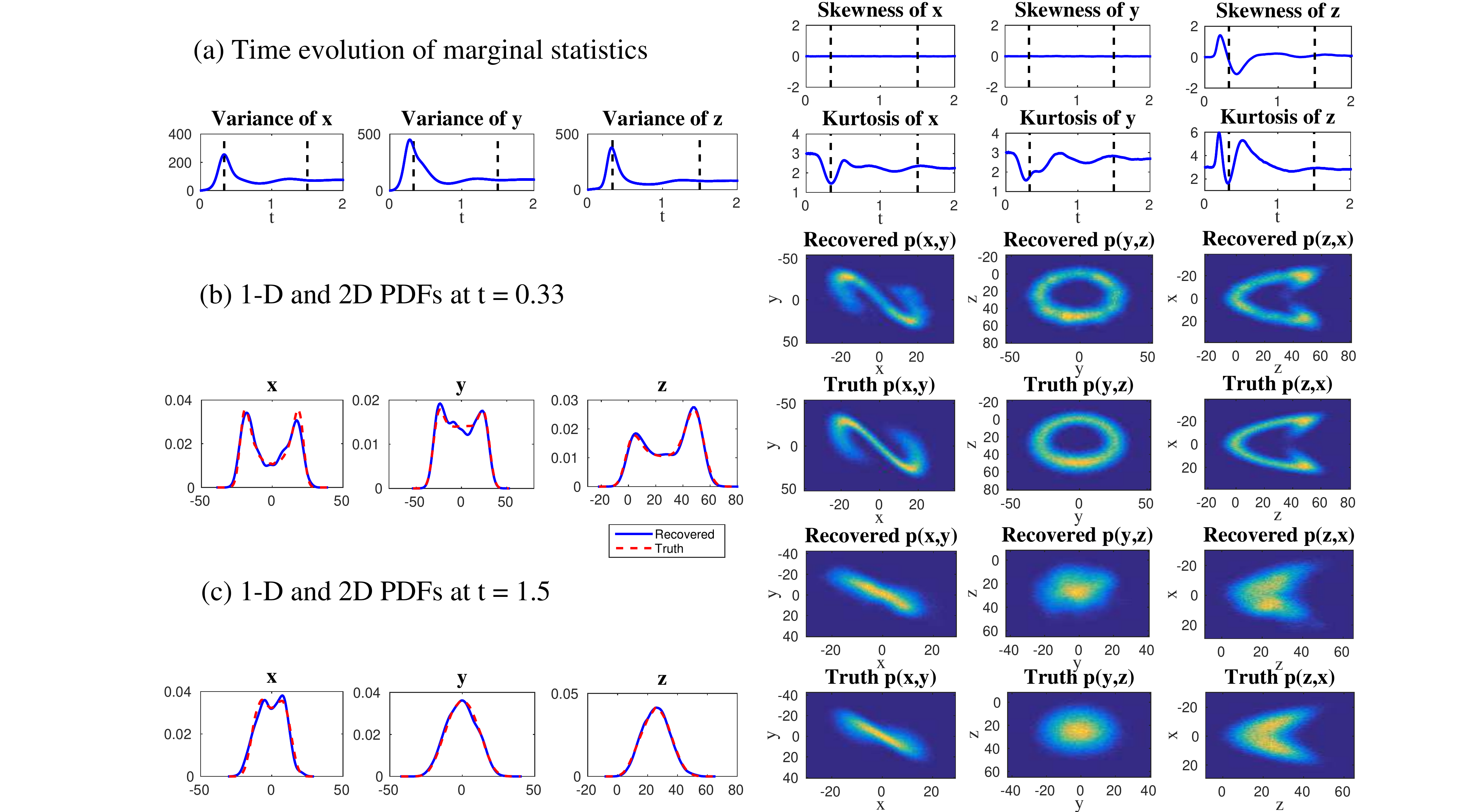}}
\caption{\textbf{L-63 model}. (a): 1D Marginal variance, skewness and kurtosis of each variable.  (b): 1D and 2D PDFs at a transition phase $t=0.33$.  (c): 1D and 2D PDFs at the statistical equilibrium phase $t=1.5$. In (b) and (c), $L=500$.}\label{L63_all}
\end{figure}

\begin{figure}[!h]
{\hspace*{-3cm}\includegraphics[width=18cm]{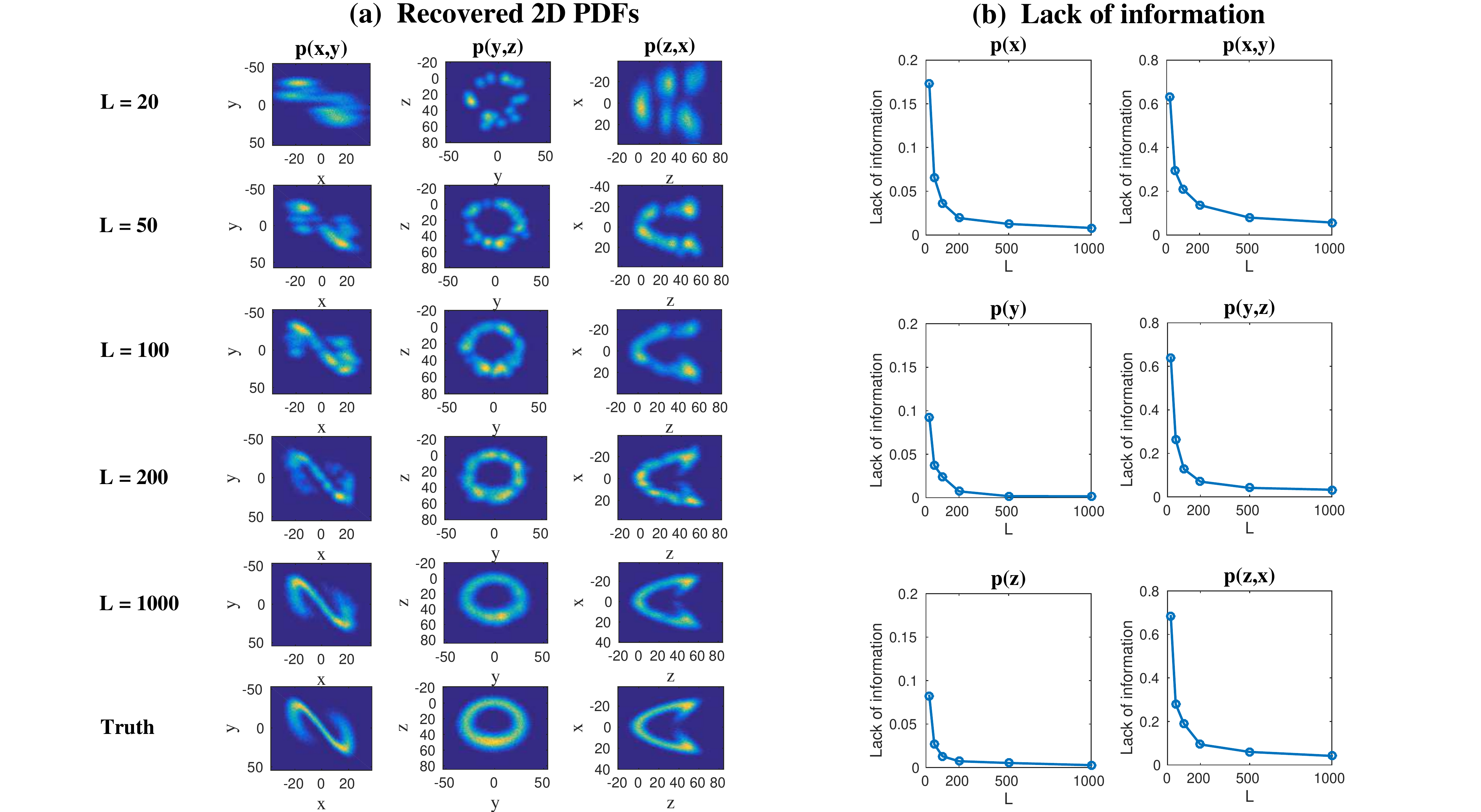}}
\caption{\textbf{L-63 model}. (a): Recovered 2D PDFs at $t=0.33$ with different $L$, and comparing with the truth. (b): The lack of information \eqref{Relative_Entropy} in the recovered 1D and 2D recovered PDFs related to the truth at $t=0.33$ as a function of $L$.}\label{L63_ModelError}
\end{figure}

\begin{figure}[!h]
{\hspace*{-3cm}\includegraphics[width=18cm]{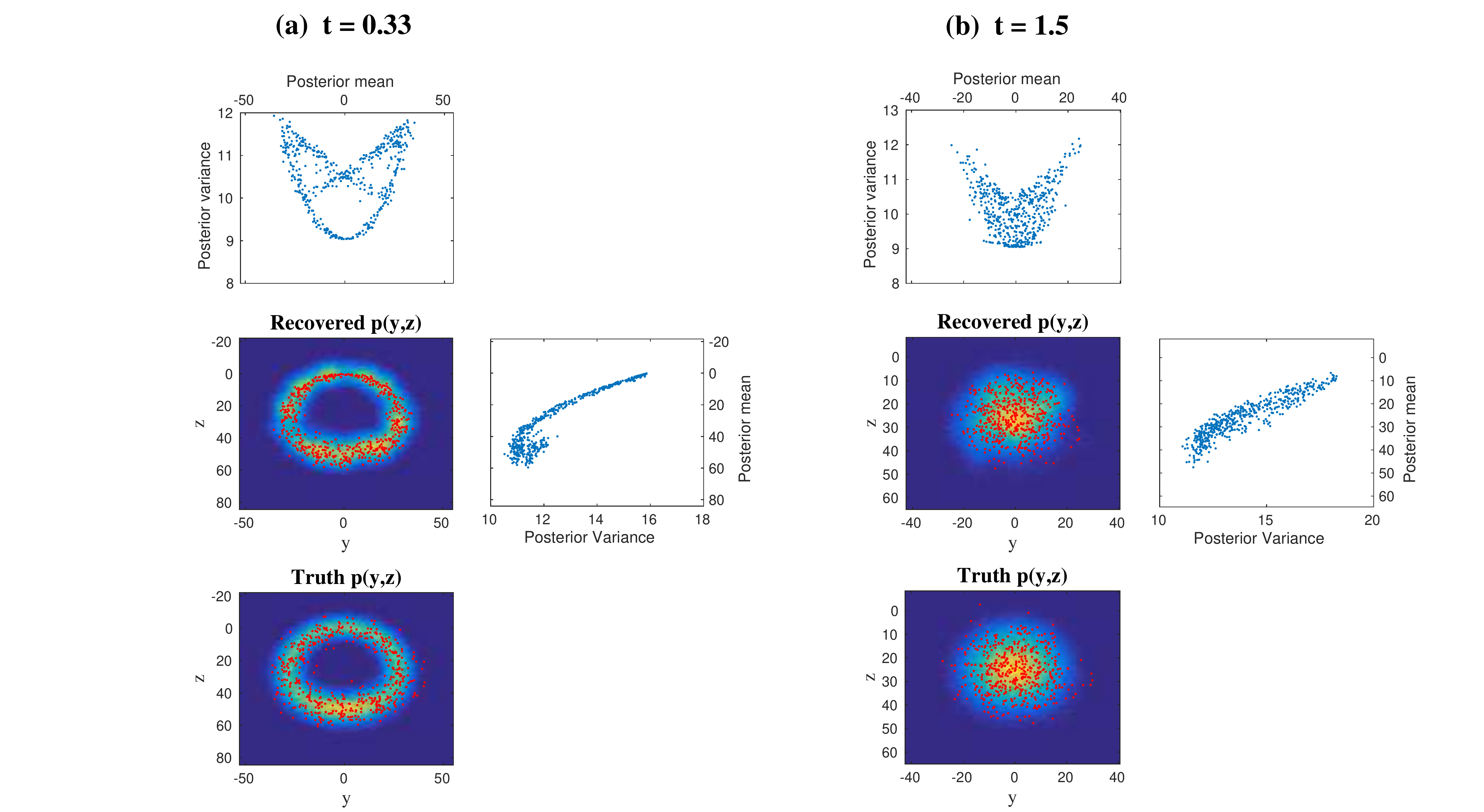}}
\caption{\textbf{L-63 model}. The red dots on top of the recovered PDF $p(y,z)$ shows the locations of $L=500$ posterior mean while those on top of the truth $p(y,z)$ are the Monte Carlo points. In addition, there are $L=500$ dots in the two panels on the top and the right sides of the recovered PDF $p(y,z)$. Each dot shows a 1D marginal posterior mean and the corresponding marginal 1D posterior variance of $y$ (top) and $z$ (right), respectively.}\label{L63_MeanVar}
\end{figure}\clearpage

\subsection{The 4D stochastic climate model (Equation \eqref{4D_Test_Model})}

In this 4D stochastic climate model, $\mathbf{u}_\mathbf{I} = (x_1,x_2)^T$ and $\mathbf{u}_\mathbf{II} = (y_1,y_2)^T$. A Gaussian initial distribution with zero mean and a diagonal covariance matrix with each diagonal entry equal to $0.1$ is adopted. Panel (a) of Figure \ref{4D_all} shows the time evolution of the skewness and kurtosis for all the four variables. The algorithm is tested in two time instants: 1) a transient phase $t=0.5$ with maximum kurtosis for $x_1$ and $y_1$ and 2) a nearly statistical equilibrium phase $t=4$. The recovered 1D and 2D PDFs are shown in Panel (b) and (c) at these two phases, respectively, with $L=500$. To illustrate the skill of recovering the tail probabilities, the comparison of the recovered 1D PDFs with the truth in the logarithm scale is also included. Clearly, the fat tails at both $t=0.5$ and $t=4$ for $x_1$ and $y_1$ are recovered accurately by the algorithm. In addition, the recovered 2D non-Gaussian PDFs and the truth also look nearly identical.

Figure \ref{4D_ModelError} shows the lack of information as a function of $L$ in the recovered 1D and 2D PDFs related to the truth. Again, the lack of information in the recovered PDF decays in an exponential fast rate and is already insignificant with  $L=100$. Note that the curve of the lack of information is similar in both the transient phase and the nearly statistical equilibrium phase, which indicates the robustness of the algorithm in recovering the PDFs at different time instants.

Finally, a comparison between using the kernel density estimation and the direct Monte Carlo in recovering the PDFs of the observed variables $x_1$ and $x_2$ is shown in Figure \ref{4D_Comparison}. With $L=100$, the kernel density estimation already succeeds in providing an accurate estimation of the 1D PDFs for both $x_1$ and $x_2$ while the histograms based on the direct Monte Carlo simulation are not even smooth, especially at the tails. It is until $L$ reaching $1000$ that the Monte Carlo simulation is able to produce a reasonably good estimation of the 1D PDFs. Nevertheless, the Monte Carlo simulation with $L=1000$ is still far from sufficient for recovering the 2D PDF while the recovered 2D PDF using the kernel density estimation is quite accurate.

\begin{figure}[!h]
{\hspace*{-3cm}\includegraphics[width=18cm]{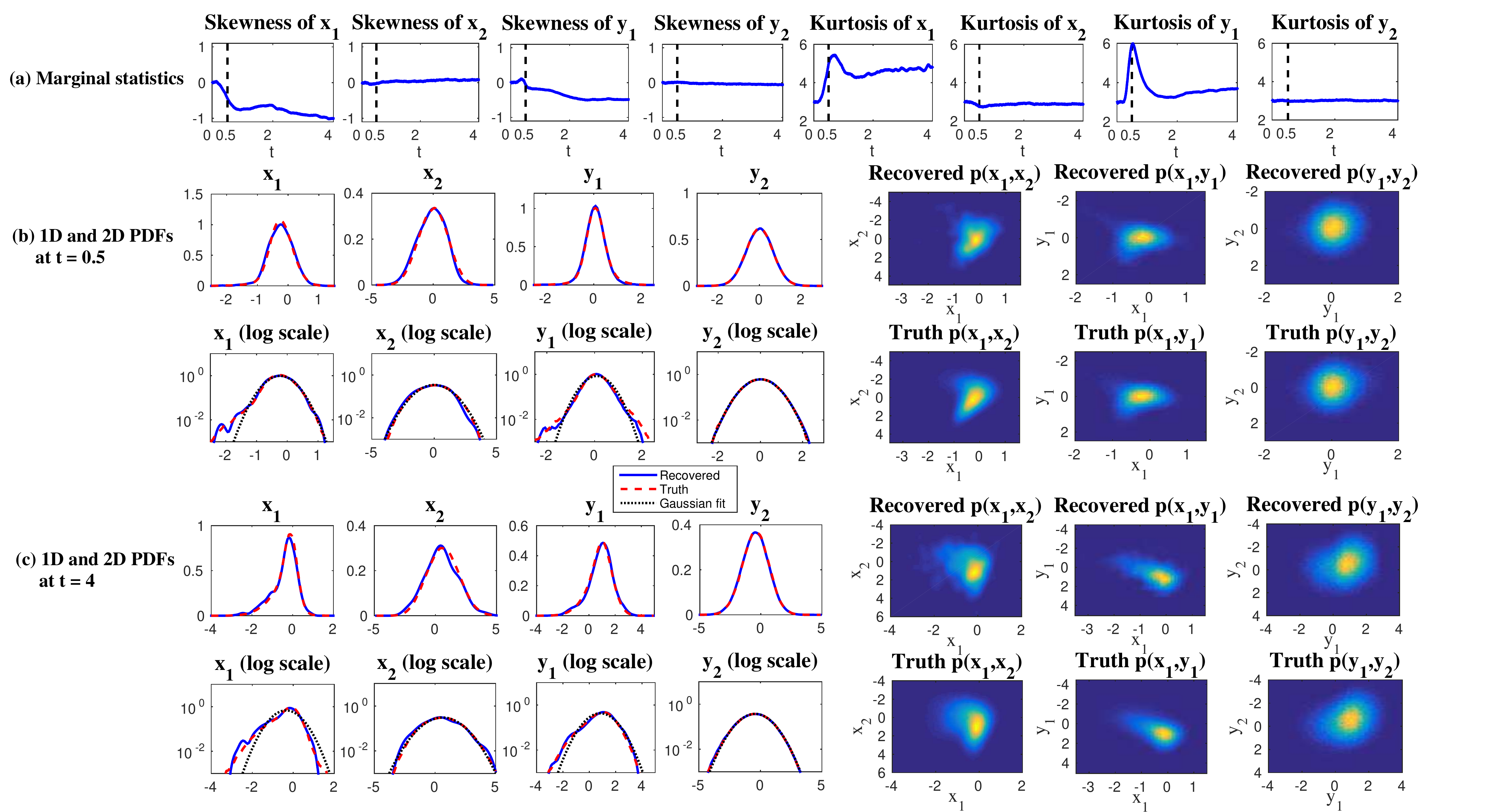}}
\caption{\textbf{4D stochastic climate model}. (a): 1D marginal skewness and kurtosis of each variable. (b): 1D and 2D PDFs at a transition phase $t = 0.5$. (c): 1D and 2D PDFs at a nearly statistical equilibrium phase $t = 4$. In (b) and (c), $L=500$.}\label{4D_all}
\end{figure}

\begin{figure}[!h]
{\hspace*{-3cm}\includegraphics[width=18cm]{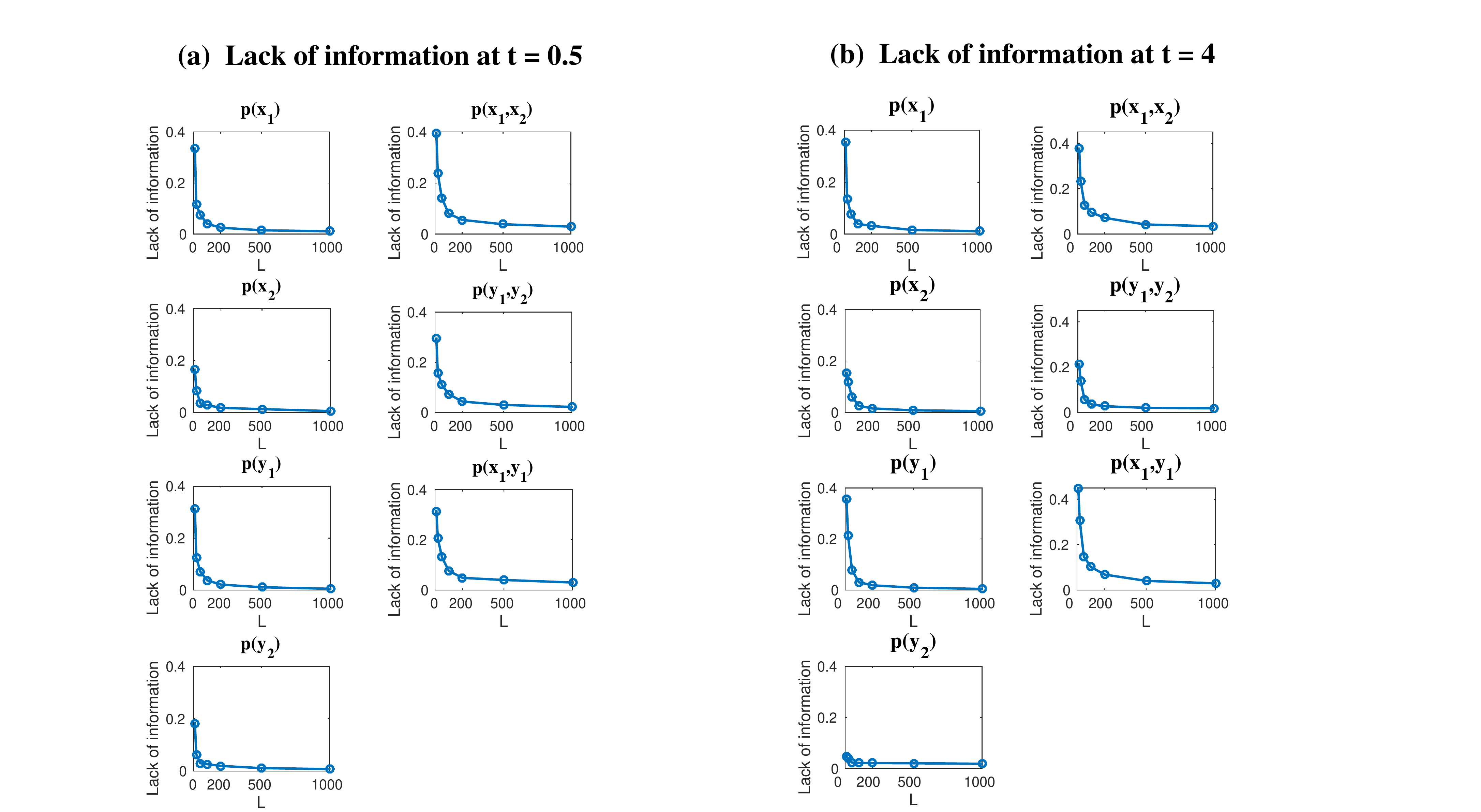}}
\caption{\textbf{4D stochastic climate model}. The lack of information \eqref{Relative_Entropy} in the recovered 1D and 2D PDFs related to the truth as a function of $L$. }\label{4D_ModelError}
\end{figure}

\begin{figure}[!h]
{\hspace*{-3cm}\includegraphics[width=18cm]{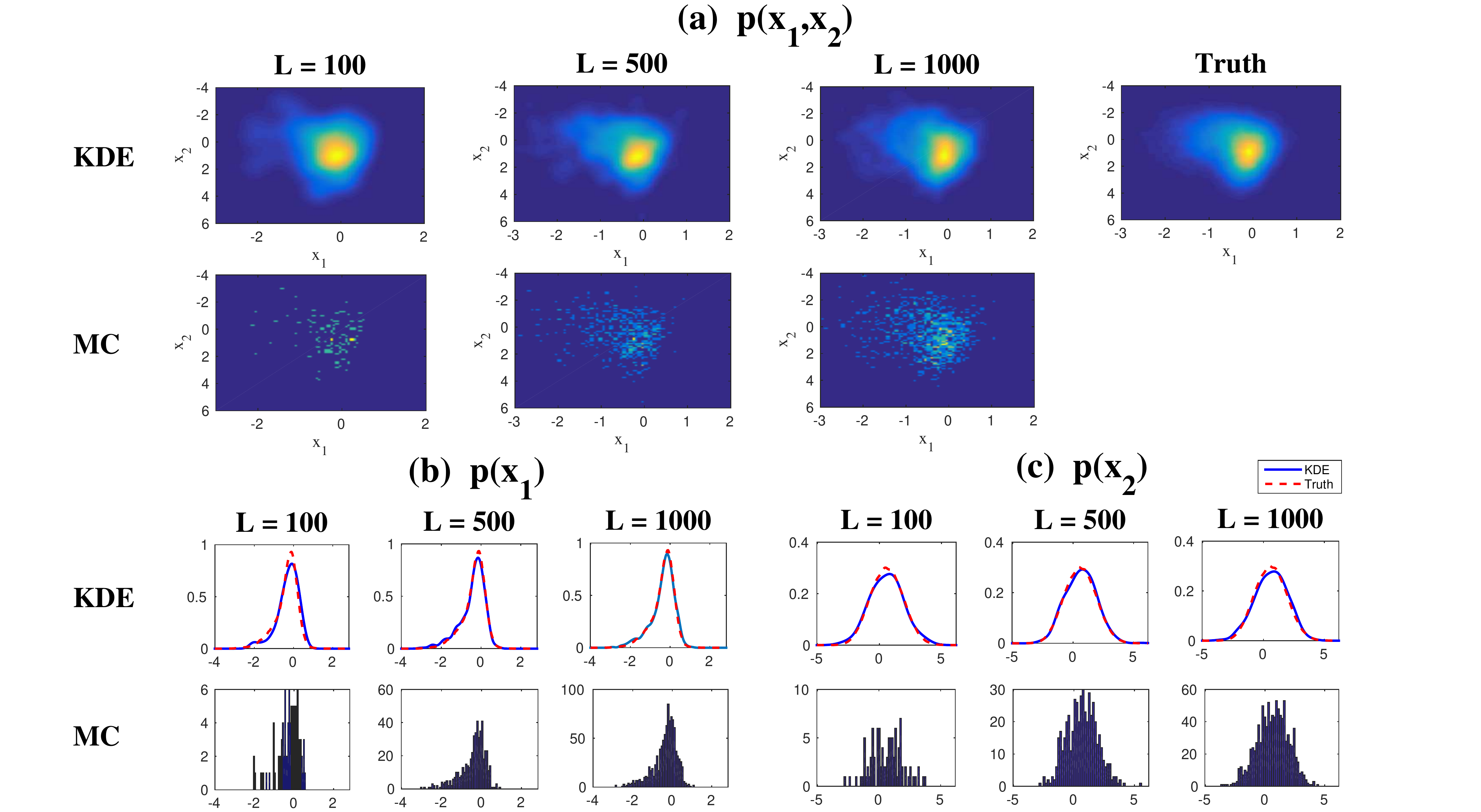}}
\caption{\textbf{4D stochastic climate model}. Comparison of the recovered 2D PDF $p(x_1,x_2)$ and 1D PDFs $p(x_1)$ and $p(x_2)$ using kernel density estimation (KDE) and MC simulation with the same $L$. }\label{4D_Comparison}
\end{figure}\clearpage

\subsection{The 3D nonlinear triad system (Equation \eqref{Linear3D})}

As discussed in Section \ref{Sec:Systems}, the 3D nonlinear triad system includes one observed variable $\mathbf{u}_{\mathbf{I}} = u_1$ and two unobserved variables $\mathbf{u}_{\mathbf{II}} = (u_2, u_3)$. The skill of recovering the non-Gaussian PDFs in the three dynamical regimes discussed in Section \ref{Sec:Systems} are shown here. See Figure \ref{3D_all_RG1}--\ref{3D_all_RG3}. The initial values of the tests in these three figures are all Gaussian PDFs centered at $(0.5,1,1)$ with a diagonal covariance with diagonal entries equal to $0.1$. Figure \ref{3D_all_RG1} shows the recovered PDFs at a transient phase and a nearly statistical equilibrium phase while Figure \ref{3D_all_RG2} and \ref{3D_all_RG3} show those at two different phases within one period. It is clear that all the non-Gaussian 1D PDFs with fat tails are reproduced by the algorithms with high accuracy using only $L=500$. The banana-shaped 2D PDF $p(u_1,u_2)$ with long tails in all the three regimes and the strongly correlated 2D PDF $p(u_2,u_3)$ in Regime III are both almost perfectly recovered as well.

Figure \ref{3D_all_RG3_Stat} compares the statistics that is recovered by the efficient statistically accurate algorithm with the truth in the toughest regime III. The recovered time evolutions of the 1D marginal mean and variance for all the three variables are almost overlapped with the truth. Even the skewness of $u_2$ and the skewness and kurtosis of $u_3$ are recovered with only small errors. The recovered skewness and kurtosis of $u_1$ is more noisy than the truth but the time-periodic trend is captured by the kernel estimation method. Importantly, the cross-covariance between all the three variables are reproduced with high skill, which justifies the block diagonal covariance matrix used in each mixture component \eqref{MeanCov} since the cross-correlation in each mixture component is already included in the conditional distribution \eqref{U_II_Mixture} and the overall cross-correlation also depends on the component locations.

Figure \ref{3D_RG2_ModelError} shows the lack of information in the recovered PDFs as a function of $L$ in Regime II, where the lack of information reduces to an insignificant amount when $L\sim O(100)$. The similar results are found in the other two regimes and are thus omitted here.

Figure \ref{3D_MeanVar} is similar to Figure \ref{L63_MeanVar} that illustrates the posterior mean and variance of each mixture component. Nevertheless, many more fascinating phenomena are revealed here. First, it is noticeable in Columns (a) and (b) that the posterior mean of the 2D distribution $p(u_2,u_3)$ stays almost in a 1D curve, which indicates that the efficient algorithm developed here involves an automatical dimension reduction process for determining the centers of the ensembles. For example, in panel (a), the posterior mean is distributed only in the $u_2$ direction and the corresponding variance of all the components is sufficiently large (right sub-panel) that is able to span the probability space of $u_3$. Secondly, a large discrepancy is likely to appear in the values of the posterior variance in different components. Looking again at panel (a), the posterior variance of  $u_2$ (top sub-panel) is much larger at the locations where $u_2$ is more negative. In fact, the marginal distribution of $u_2$ is skewed with a one-side fat tail towards the negative direction (Figure \ref{3D_all_RG1}). The large values of the posterior variance in this tail region implies that a substantial amount of area is covered by each conditional Gaussian distribution and therefore a small number $L$ is sufficient for an accurate estimation of this fat tail. This is a striking advantage over the Monte Carlo simulation that usually requires a large number of samples to simulate the fat tails.

Finally, we test the algorithm at a short transient phase $t=0.05$ starting from highly non-Gaussian initial values with large variance in the toughest regime III. The initial distribution of $u_2$ is assumed to be either a Gamma distribution or a bimodal distribution. See the left column of Figure \ref{3D_all_RG3_Initial}. Despite the uncorrelated initial distributions of $u_2$ and $u_3$, the strong coupling with fast oscillations between these two variables results in a significant tilted structure in the 2D PDF $p(u_2,u_3)$ at this short transient phase. The PDF recovered by the algorithm is able to capture such tilt as well as the non-Gaussian features starting from different initial values with the lack of information smaller than $0.1$ in the recovered PDFs.

\begin{figure}[!h]
{\hspace*{-3cm}\includegraphics[width=18cm]{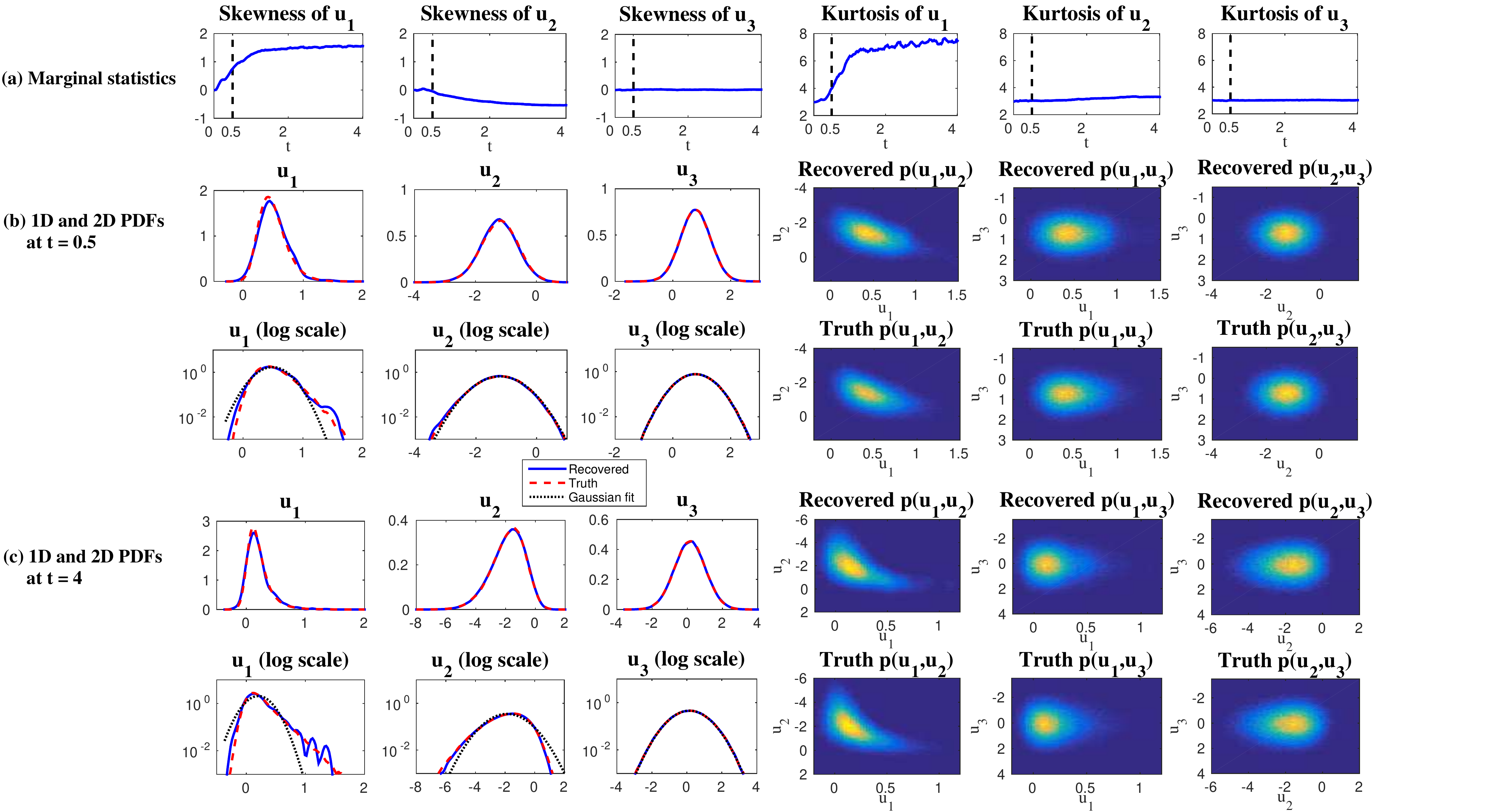}}
\caption{\textbf{3D nonlinear triad system; Regime I}. (a): 1D marginal skewness and kurtosis of each variable. (b): 1D and 2D PDFs at a transition phase $t = 0.5$. (c): 1D and 2D PDFs at a nearly statistical equilibrium phase $t = 4$. In (b) and (c), $L=500$.}\label{3D_all_RG1}
\end{figure}

\begin{figure}[!h]
{\hspace*{-3cm}\includegraphics[width=18cm]{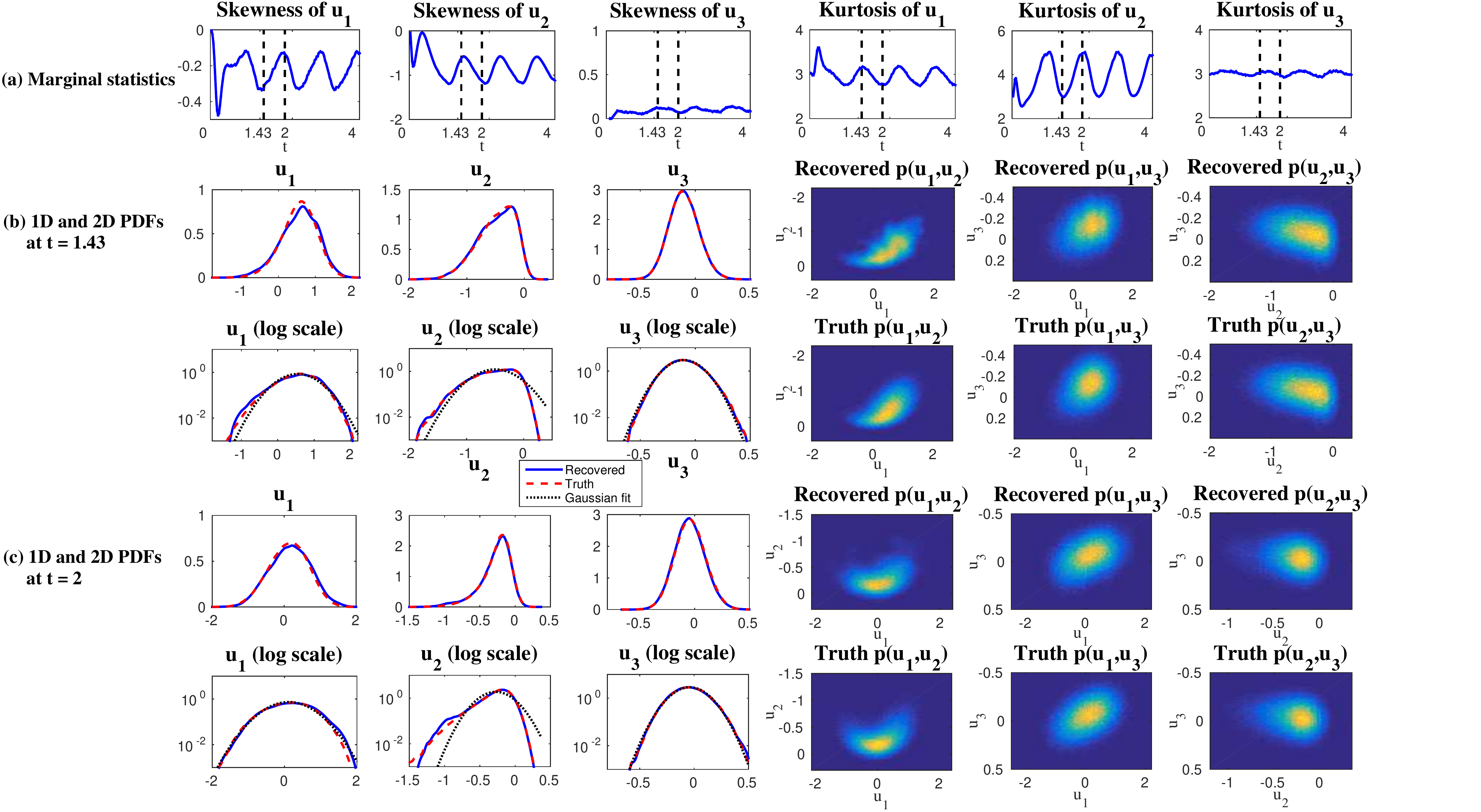}}
\caption{\textbf{3D nonlinear triad system; Regime II}. (a): Same as Figure \ref{3D_all_RG1}. (b) and (c): 1D and 2D PDFs at two transition phases $t = 1.43$ and $t=2$, respectively.}\label{3D_all_RG2}
\end{figure}

\begin{figure}[!h]
{\hspace*{-3cm}\includegraphics[width=18cm]{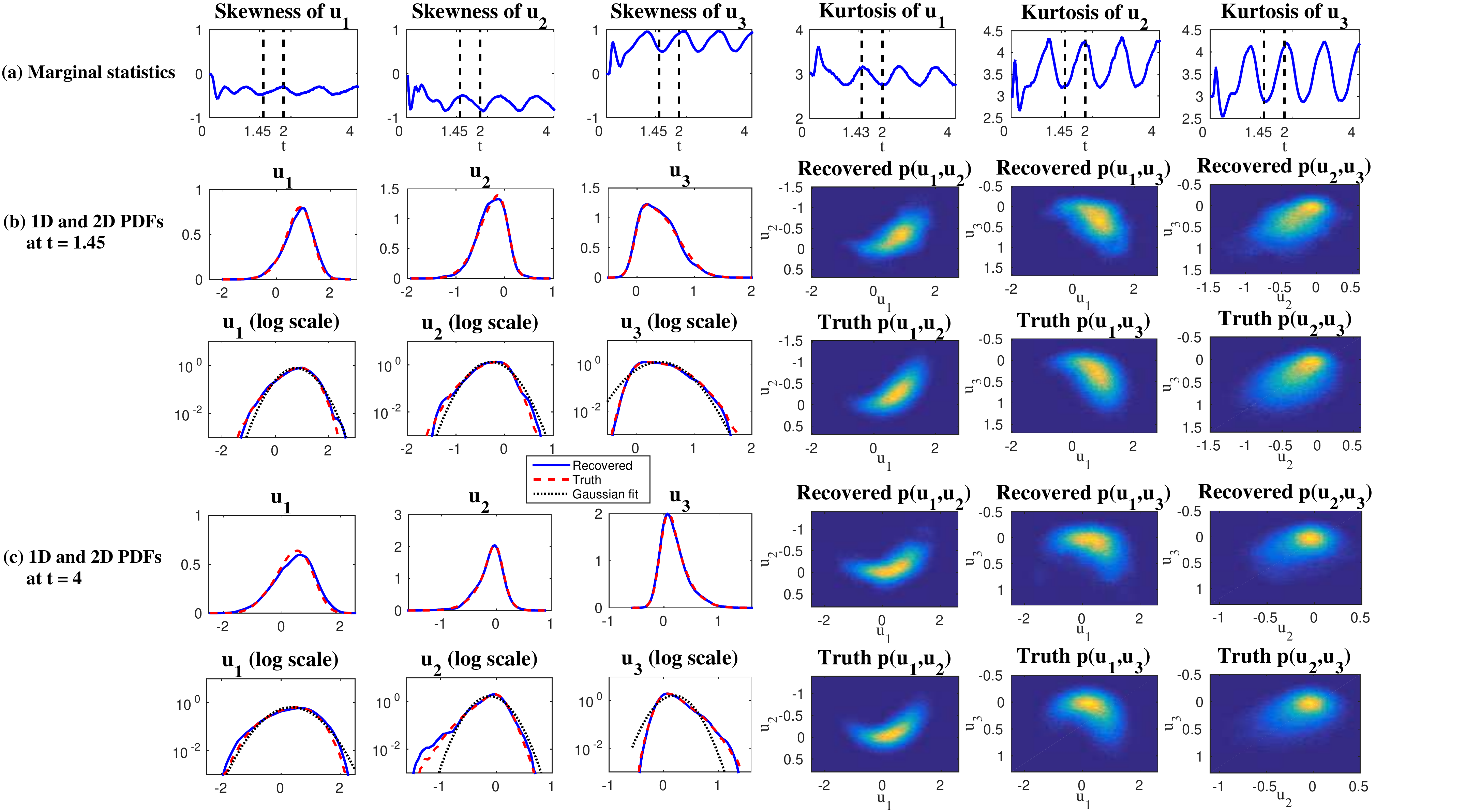}}
\caption{\textbf{3D nonlinear triad system; Regime III}. (a): Same as Figure \ref{3D_all_RG1}. (b) and (c): 1D and 2D PDFs at two transition phases $t = 1.45$ and $t=2$, respectively.}\label{3D_all_RG3}
\end{figure}

\begin{figure}[!h]
{\hspace*{-4cm}\includegraphics[width=18cm]{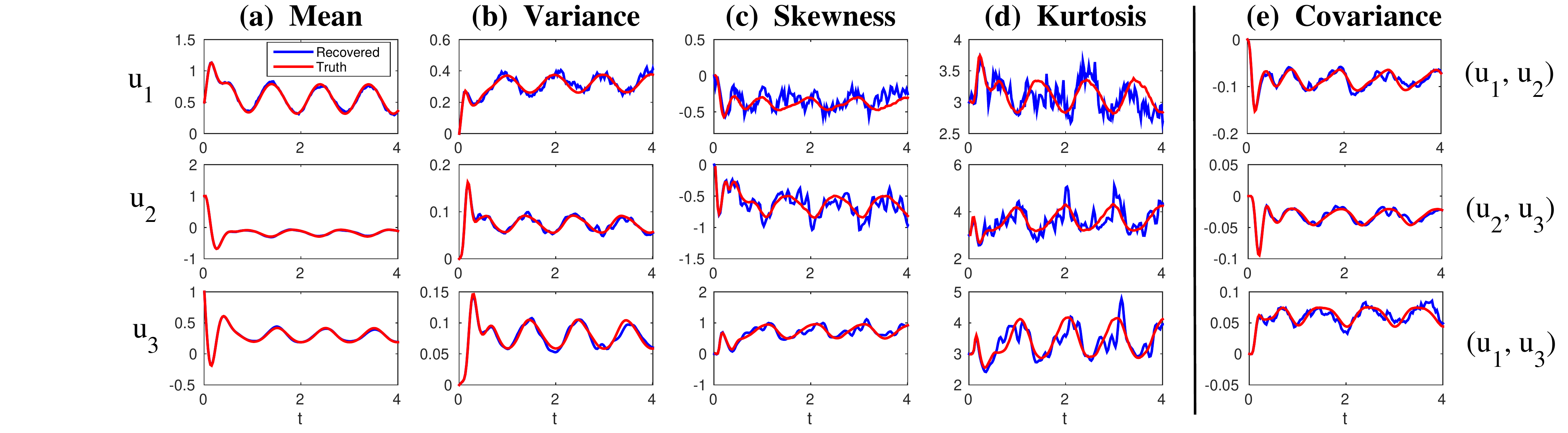}}
\caption{\textbf{3D nonlinear triad system; Regime III}. Comparison of the recovered time evolution of the statistics ($L=500$) with the truth.}\label{3D_all_RG3_Stat}
\end{figure}

\begin{figure}[!h]
{\hspace*{-3cm}\includegraphics[width=18cm]{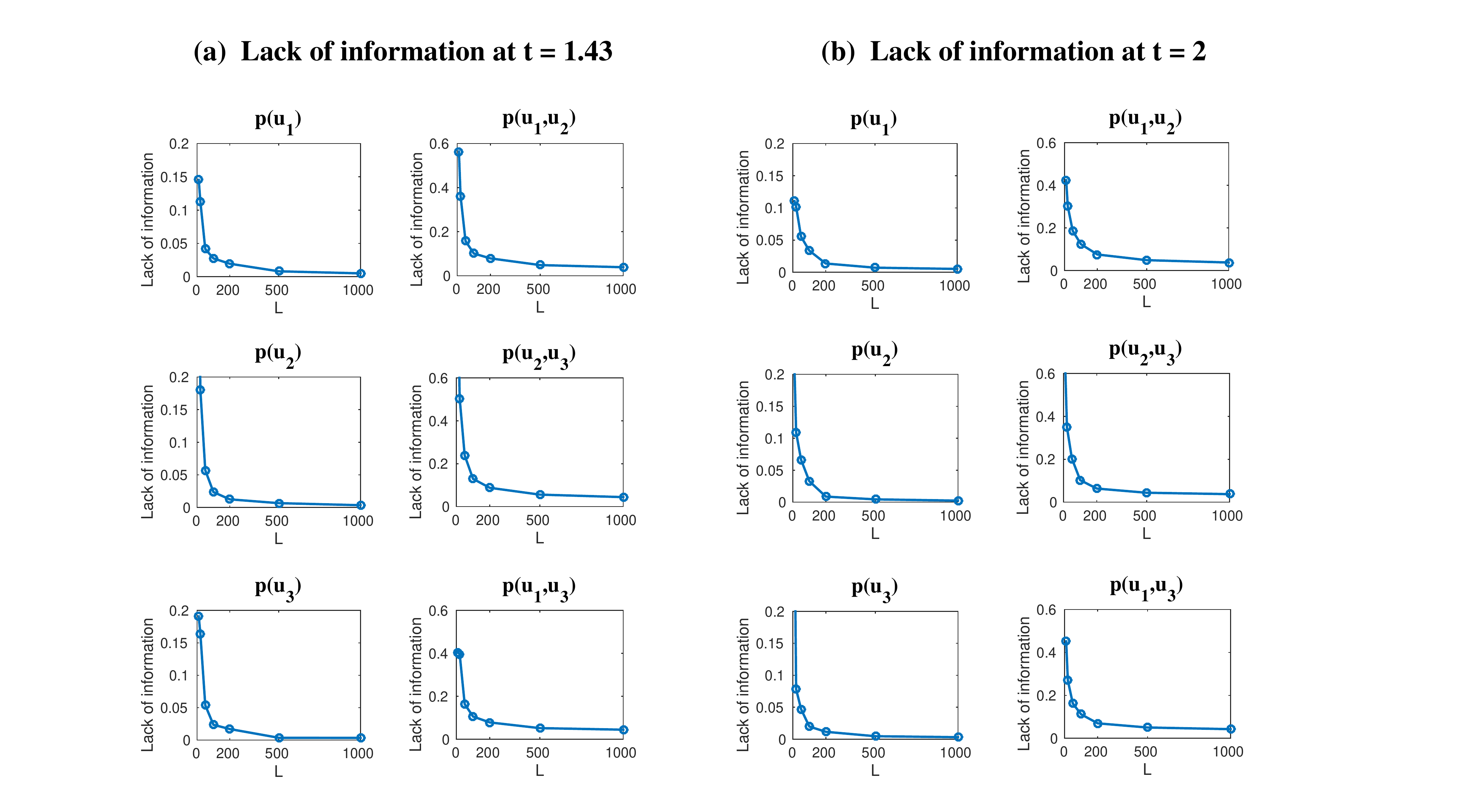}}
\caption{\textbf{3D nonlinear triad system; Regime II}. The lack of information \eqref{Relative_Entropy} in the recovered 1D and 2D PDFs related to the truth as a function of $L$.}\label{3D_RG2_ModelError}
\end{figure}

\begin{figure}[!h]
{\hspace*{-3cm}\includegraphics[width=18cm]{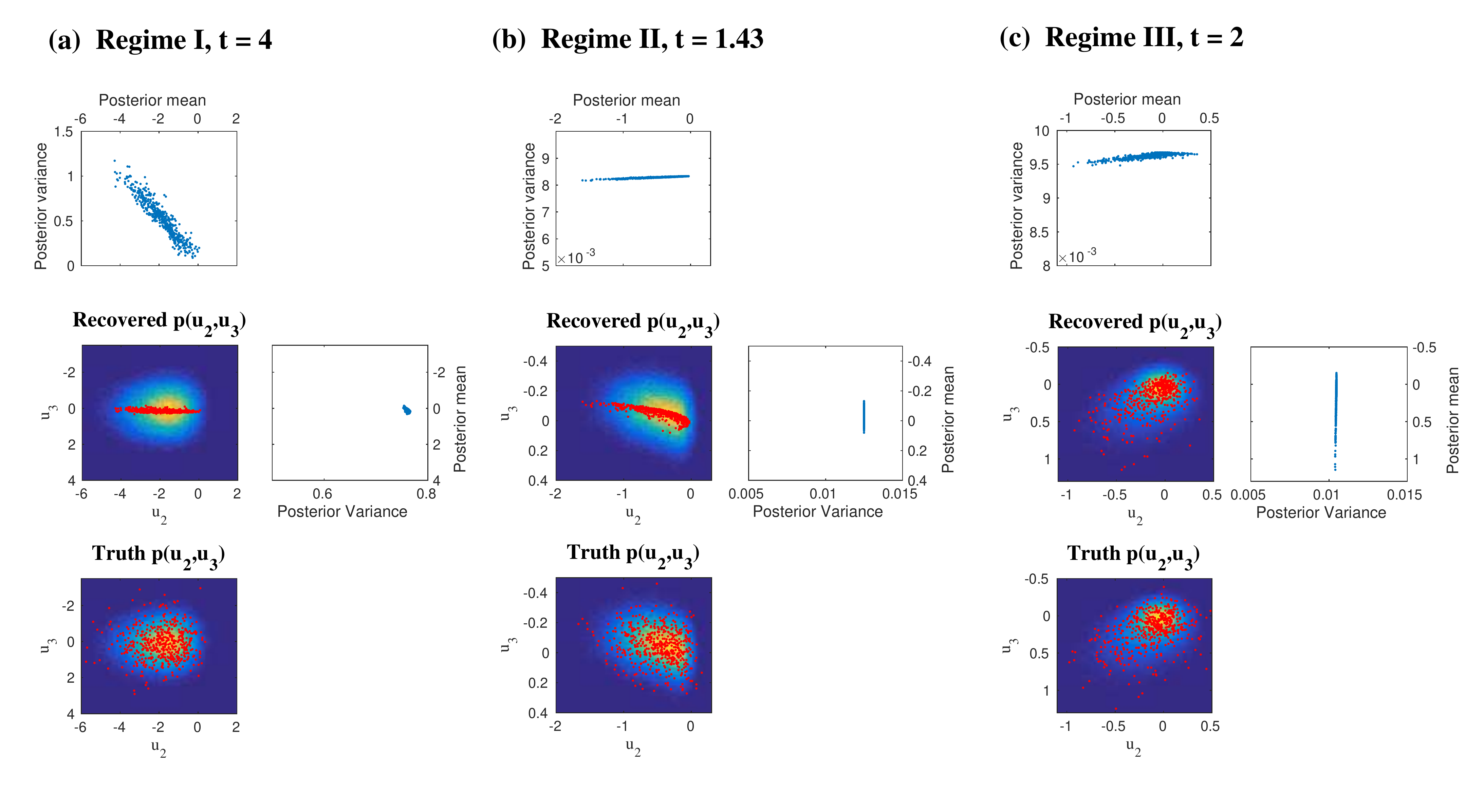}}
\caption{\textbf{3D nonlinear triad system}. The red dots on top of the recovered PDF $p(y,z)$ shows the locations of $L=500$ posterior mean while those on top of the truth $p(y,z)$ are the Monte Carlo points. There are $L=500$ dots in the two panels on the top and right side of the recovered PDF $p(y,z)$. Each dot shows a 1D marginal posterior mean and the corresponding marginal 1D posterior variance of $y$ and $z$, respectively. (a), (b) and (c) show different behaviors in the three regimes.}\label{3D_MeanVar}
\end{figure}

\begin{figure}[!h]
{\hspace*{-4cm}\includegraphics[width=18cm]{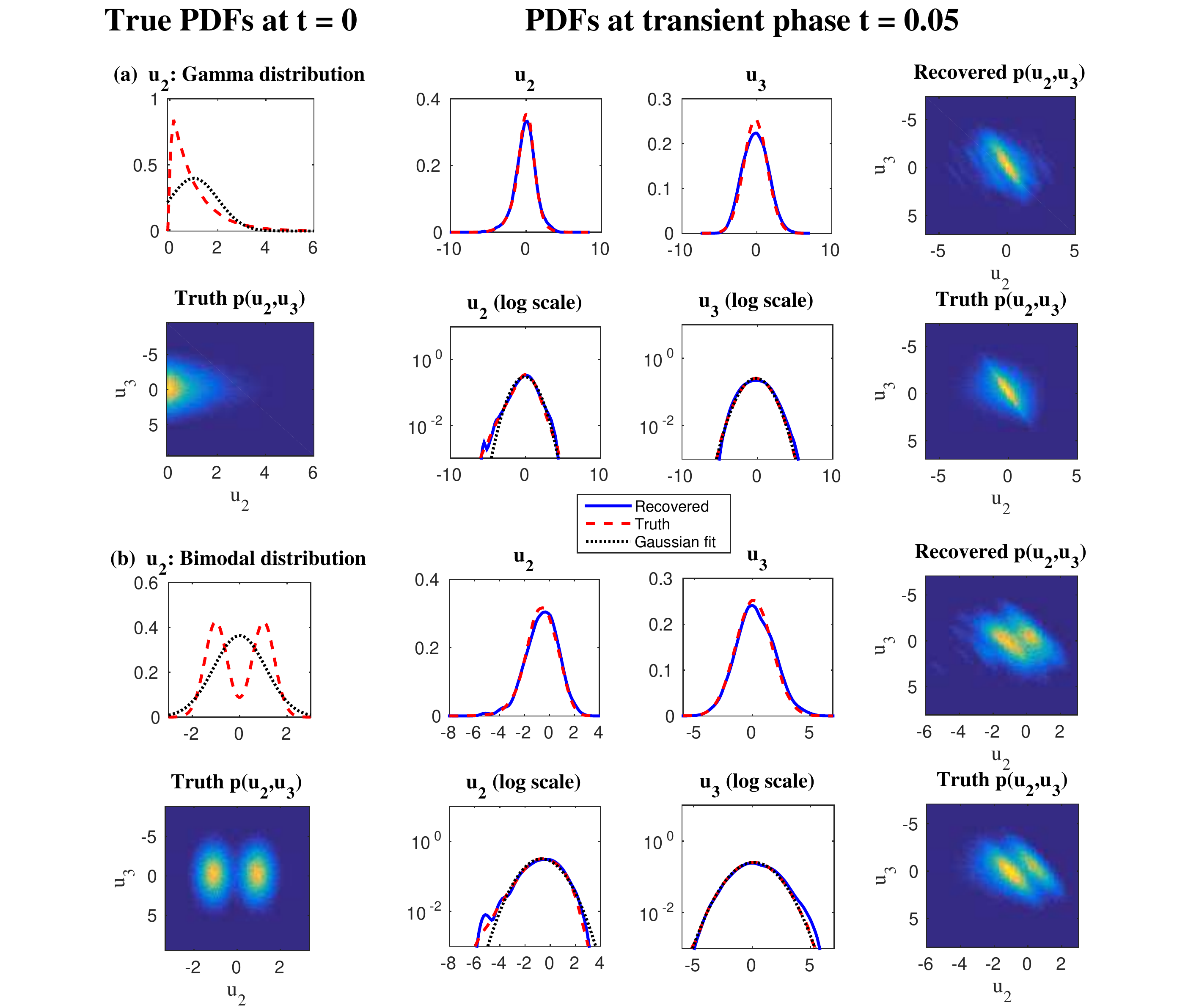}}
\caption{\textbf{3D nonlinear triad system; Regime III}. Recovery of the PDF of the unobserved variables $p(u_2,u_3)$ at a short transient phases $t=0.05$ with a large variance in the highly non-Gaussian initial values. The initial distribution of $u_2$ satisfies a Gamma distribution $\Gamma(1,1)$ in (a) and a bimodal distribution with the superposition of two Gaussians $\mathcal{N}(1,0.2)$ and $\mathcal{N}(-1,0.2)$ in (b). The initial distributions of $u_1$ and $u_3$ in both cases are $\mathcal{N}(0,4)$. Here, $L=500$.}\label{3D_all_RG3_Initial}
\end{figure}\clearpage

\subsection{The 6D conceptual dynamical model for turbulence (Equation \eqref{TurbulentModel})}

Finally, the skill of recovering the highly non-Gaussian PDFs of the 6D turbulent model is reported here, where $\mathbf{u}_\mathbf{I}=u$ and $\mathbf{u}_\mathbf{II} = (v_1,\ldots,v_5)^T$ with zero initial values for all the 6 variables. Figure \ref{6D_Regime} shows the time evolution of the 1D marginal mean, variance, skewness and kurtosis for all the variables. Note that $t=0.6$ is a transient phase at which the small-scale unobserved variables $v_3, v_4$ and $v_5$ have the strongest non-Gaussian features with both large skewness and kurtosis.

Figure \ref{6D_all_t06} compares the recovered 1D and 2D PDFs with the truth at this transient phase $t=0.6$ and Figure \ref{6D_all_t4} shows those at the nearly statistical equilibrium phase $t=4$. It is clear that $L=500$ is sufficient to recover the 1D skewed PDFs with an one-side fat tail associated with the small-scale variables as well as the Gaussian and non-Gaussian features in those medium- and large-scale variables. The efficient statistically accurate algorithm also provides an accurate estimation of all the 2D joint PDFs. Particularly, the banana shapes of the 2D PDFs in $p(u,v_i)$ and the strong correlations between $v_i$ and $v_j$ in $p(v_i,v_j)$ are both reproduced with high accuracy. The lack of information as a function of $L$ in the recovered 1D and 2D PDFs related to the truth at $t=4$ is shown in Figure \ref{6D_ModelError}. This is similar to those in all the previous test models, indicating the robustness of the algorithm in recovering the PDFs for various turbulent systems at different phases.

\begin{figure}[!h]
{\hspace*{-3cm}\includegraphics[width=18cm]{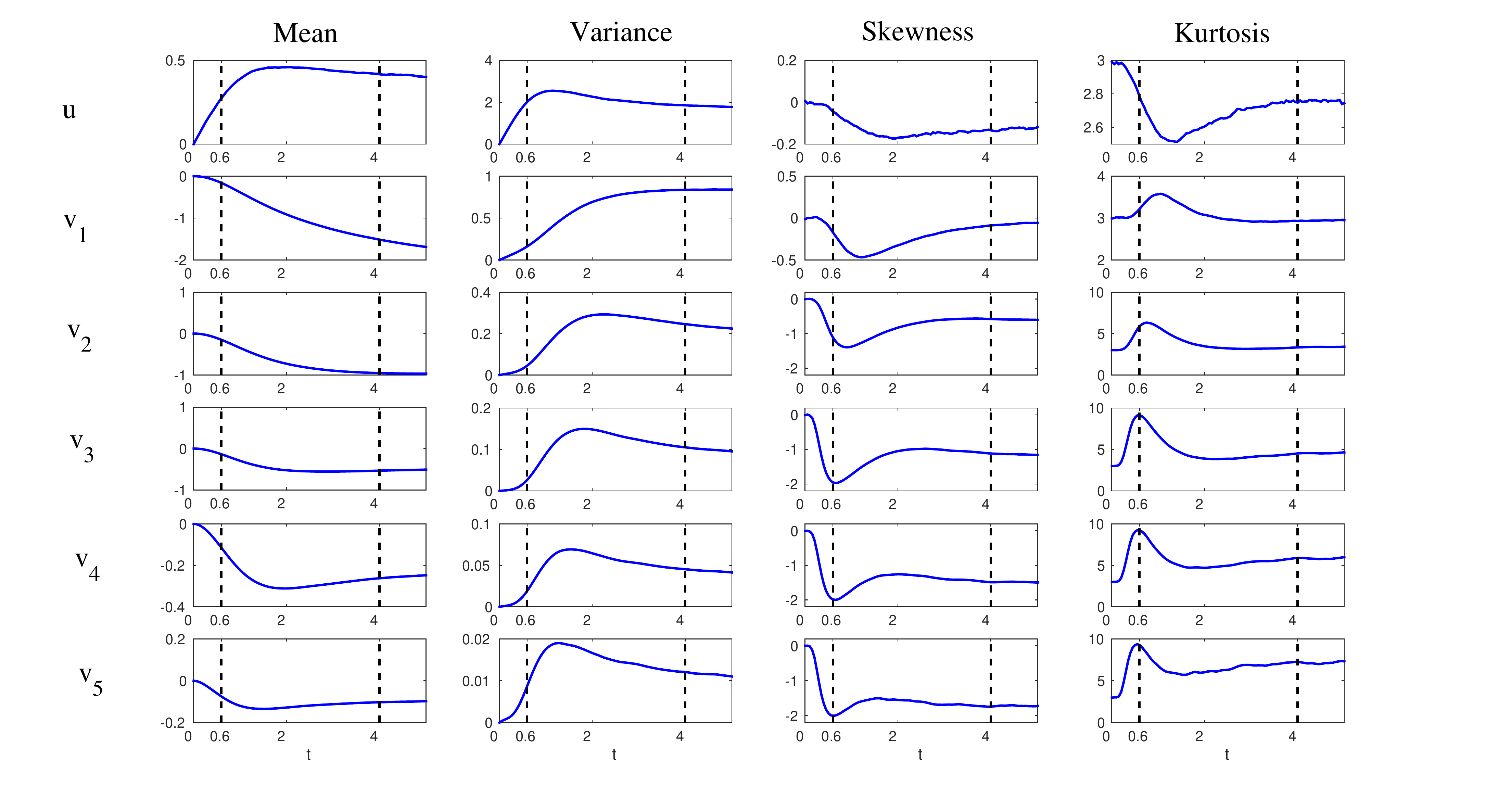}}
\caption{\textbf{6D conceptual dynamical model for turbulence}. Evolutions of 1D marginal statistics.}\label{6D_Regime}
\end{figure}

\begin{figure}[!h]
{\hspace*{-3cm}\includegraphics[width=18cm]{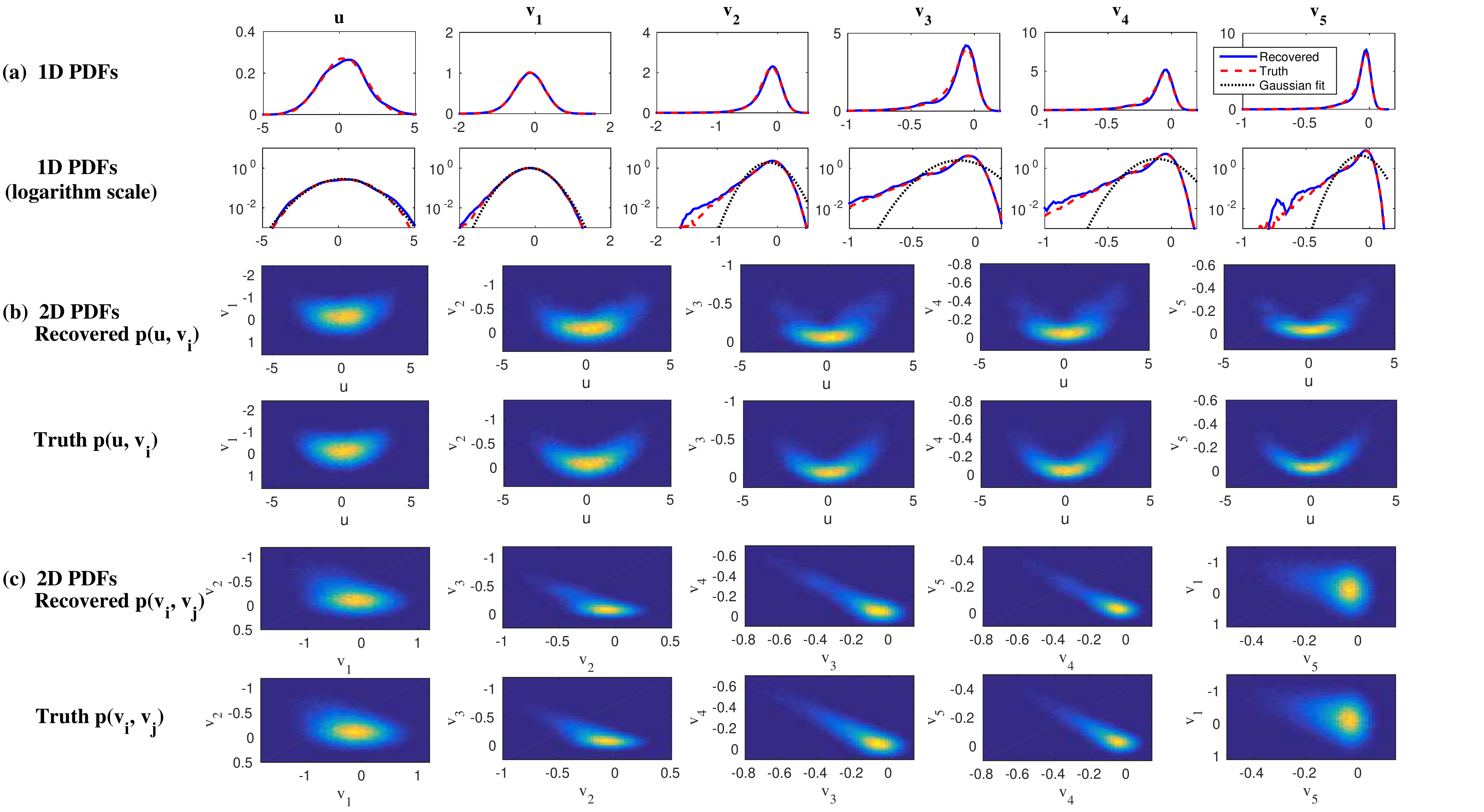}}
\caption{\textbf{6D conceptual dynamical model for turbulence}. At a transient phase $t=0.6$. (a): 1D PDFs. (b): 2D PDFs $p(u,v_i)$. (c): 2D PDFs $p(v_i,v_j)$. Here $L=500$.}\label{6D_all_t06}
\end{figure}

\begin{figure}[!h]
{\hspace*{-3cm}\includegraphics[width=18cm]{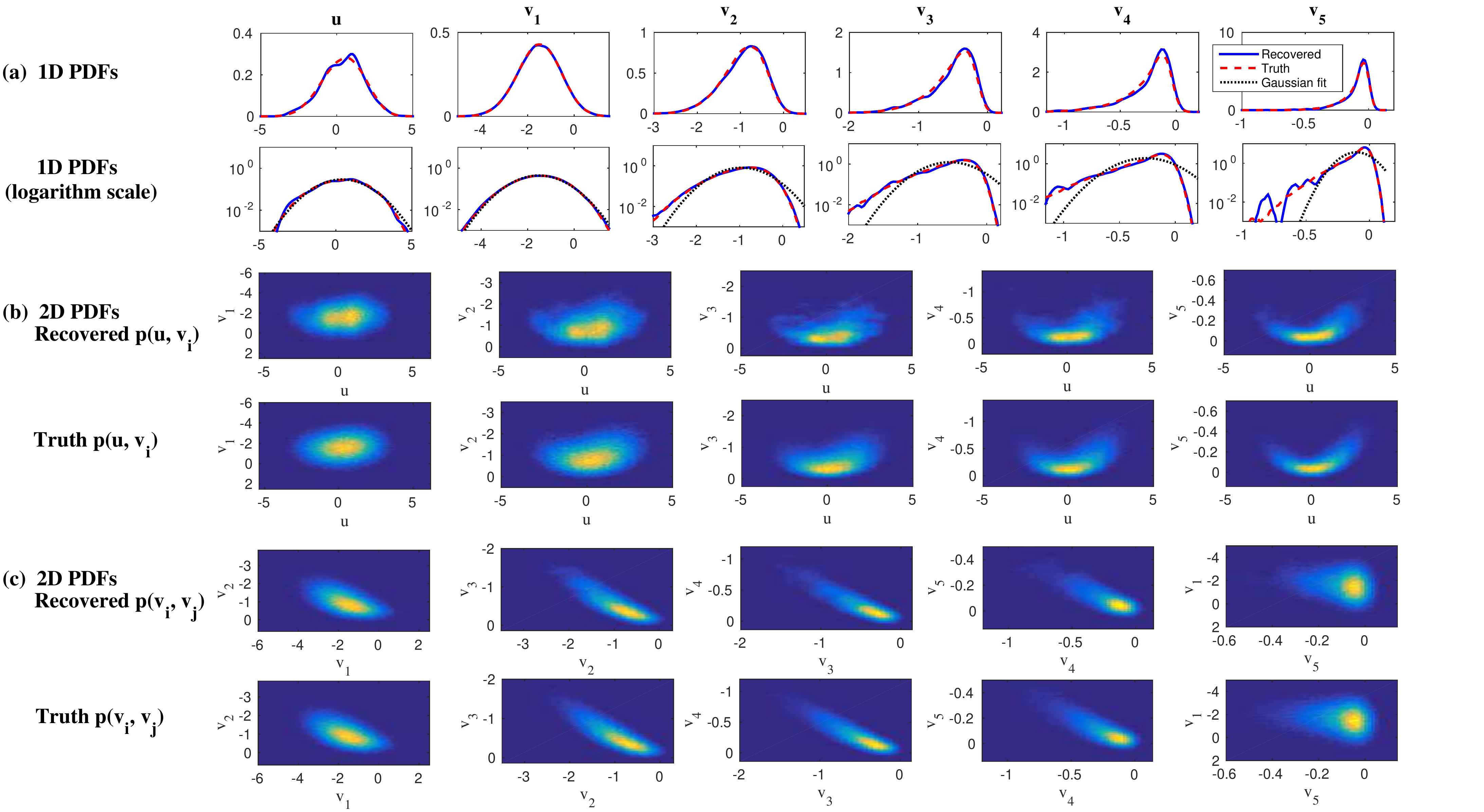}}
\caption{\textbf{6D conceptual dynamical model for turbulence}. At the nearly statistical equilibrium phase $t=4$. (a): 1D PDFs. (b): 2D PDFs $p(u,v_i)$. (c): 2D PDFs $p(v_i,v_j)$. Here $L=500$.}\label{6D_all_t4}
\end{figure}

\begin{figure}[!h]
{\hspace*{-3cm}\includegraphics[width=18cm]{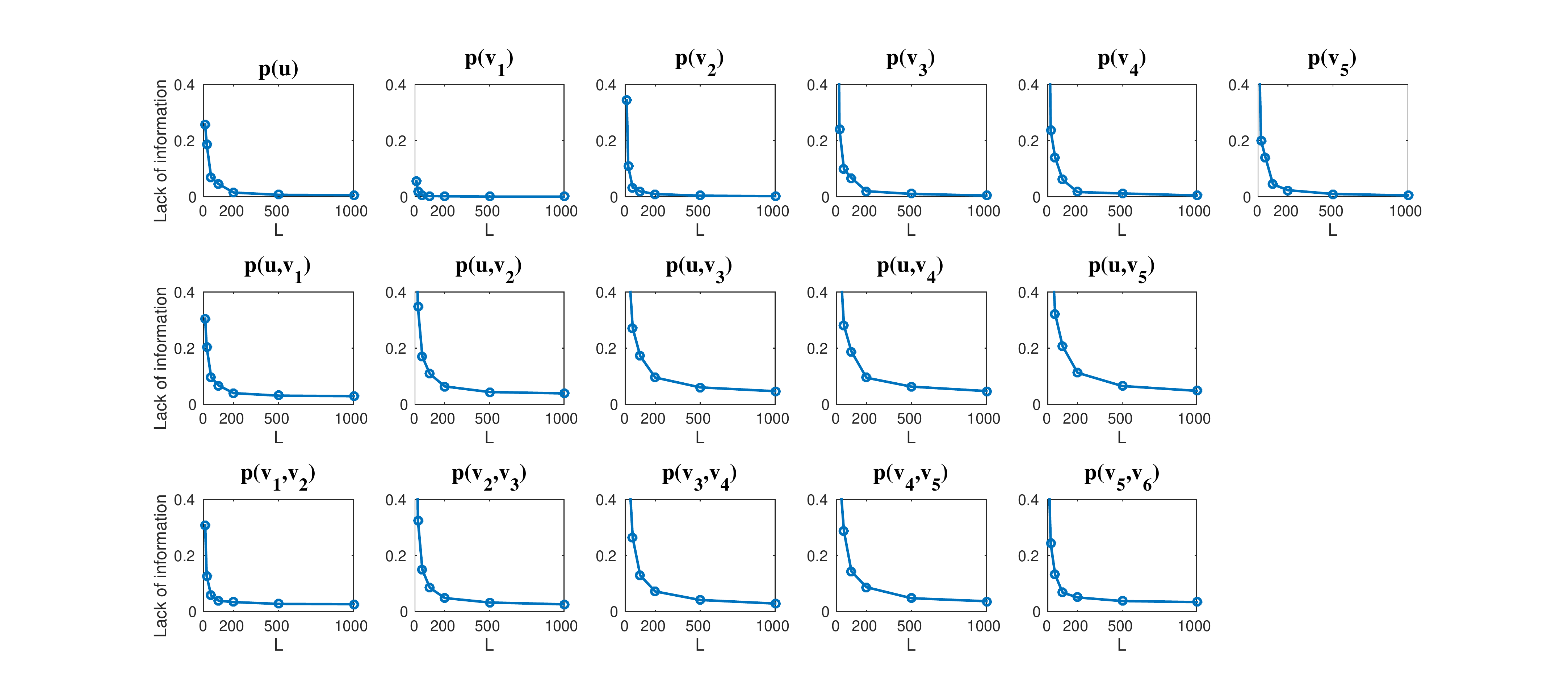}}
\caption{\textbf{6D conceptual dynamical model for turbulence}. The lack of information  in the recovered  1D and 2D PDFs related to the truth as a function of $L$ at $t=4$.}\label{6D_ModelError}
\end{figure}

\clearpage
\section{Discussion and Conclusions}\label{Sec:Conclusion}
In this article, efficient statistically accurate algorithms are developed for solving the Fokker-Planck equation associated with the conditional Gaussian turbulent dynamical systems in large dimensions \eqref{Conditional_Gaussian_System}.

Despite the conditional Gaussianity, the conditional Gaussian systems are nonlinear and can be highly non-Gaussian in both transient phases and the statistical steady state. They are able to capture many desired characteristics of turbulence, neuroscience and excitable media. In particularly, the conditional Gaussian framework includes a rich class of the turbulent models that contain energy-conserving quadratic nonlinear interactions as in nature \citep{majda2016introduction} (Section \ref{Sec:Systems} and \ref{Appendix:A}). One important feature of the conditional Gaussian systems is that the conditional distribution $p(\mathbf{u}_\mathbf{II}(t)|\mathbf{u}_\mathbf{I}(s\leq t))$ of the unobserved variables $\mathbf{u}_\mathbf{II}(t)$ given each trajectory of the observed variables $\mathbf{u}_\mathbf{I}(s\leq t)$ is Gaussian \eqref{CG_PDF} and it can be solved via closed analytical formulae \eqref{CG_Result}.

The efficient statistically accurate algorithms developed here involve a hybrid strategy. The PDF of the high-dimensional unobserved variables $p(\mathbf{u}_\mathbf{II}(t))$ is recovered by a parametric method that employs the average of $L$ conditional Gaussian posterior distributions $p(\mathbf{u}_\mathbf{II}(t)|\mathbf{u}^i_\mathbf{I}(s\leq t)), i=1,\ldots, L$ (Proposition \ref{Prop:Unobserved}). On the other hand, the PDF of the low-dimensional observed variables $p(\mathbf{u}_\mathbf{I}(t))$ is solved using a judicious non-parametric kernel density estimation method with Gaussian kernels (Proposition \ref{Prop:Observed}). The combination of the $L$ Gaussian distributions for $p(\mathbf{u}_\mathbf{I}(t))$ and the $L$ conditional Gaussian distributions for $p(\mathbf{u}_\mathbf{II}(t))$ leads to a Gaussian mixture for recovering the joint distribution $p(\mathbf{u}_\mathbf{I}(t), \mathbf{u}_\mathbf{II}(t))$ (Proposition \ref{Prop:Both}). In the limit $L\to\infty$, the solution resulting from these algorithms is consistent with that of solving the Fokker-Planck equation \eqref{Fokker_Planck_Equation}. Practically, $L\sim O(100)$ is able to provide an accurate estimation of  non-Gaussian PDFs with dimension $\sim O(10)$. The success of solving the high dimensional PDF with such a small number of mixture components is due to the sufficiently large portion of the high-dimensional PDF being covered by each component, which is completely different from the traditional particle methods. As shown in Figure \ref{L63_MeanVar} and \ref{3D_MeanVar} in the performance tests, the marginal variances of each conditional Gaussian distribution associated with the high dimensional unobserved variables are usually large, and therefore the portion consisting of an enormous number of Monte Carlo samples can be covered by only one conditional Gaussian component. In particular, the posterior variances become even more significant when the associated mixture component is located in the fat tails of the distribution (e.g., Panel (a) of Figure \ref{3D_MeanVar}). This is another advantage of the algorithm in improving the efficiency of capturing the non-Gaussian features resulting from the intermittency and extreme events. In addition, the algorithms sometimes also involve an automatic dimension reduction process that makes the posterior mean locate in a lower dimensional subspace (e.g., Panel (a) and (b) of Figure \ref{3D_MeanVar}), which further reduces the number of mixture components required in the algorithms. All the properties presented above play important roles in ameliorating the curse of dimensionality and facilitate the algorithms to deal with high dimensional PDFs with strong non-Gaussian features. We also show in our test examples that the algorithms behave in a uniformly convergent fashion at long times with $t\to\infty$. Note that the posterior distributions of the high-dimensional unobserved part $\mathbf{u}_\mathbf{II}$ are solved via closed analytical formulae and the posterior distributions associated with different components can be solved in a parallel way, which greatly reduce the computational cost and avoid approximate errors.

Numerical simulations in the performance tests show that the error (lack of information) in the recovered PDFs decays exponentially fast as a function of $L$. In addition to these numerical results, a rigorous mathematical analysis of the error bound depending on different factors in the model will be very useful to understand the convergence of the efficient statistically accurate algorithms and provide guidelines for further improvement of the algorithms. Particularly, it is extremely interesting to show in a rigorous way that both the large posterior variance and the automatic dimension reduction of the posterior mean  play crucial roles in ameliorating the curse of dimensionality and allow the algorithms to work in high dimensional systems. These theoretical issues are addressed in an ongoing work \citep{chen2017rigorous}.
Note that the aim of this article is to develop these efficient statistically accurate algorithms and numerically validate their performance. The test models used here have dimensions only up to $6$. This is because obtaining the true PDF of higher dimensional systems is not a simple task for validation. An inaccurate true PDF will introduce difficulties for quantifying the error in the recovered ones. Nevertheless, given the validation tests in this work, the algorithms can be applied to turbulent dynamical systems with higher dimensions in the future. In addition, since the strategy of dealing with $\mathbf{u}_\mathbf{II}$ is already sufficiently efficient and accurate for large dimensions, improving the strategy to handle systems with a larger dimension of $\mathbf{u}_\mathbf{I}$ can be a future direction.
Finally, the efficient statistically accurate algorithms have many important applications. For example, they can be applied to study the causality between different phenomena in the atmosphere and ocean,  which involves computing the information transfer in high-dimensional turbulent dynamical systems based on the associated non-Gaussian PDFs \citep{majda2007information, liang2005information}. They can also be applied to solve the joint PDF of the turbulent ocean flows and the associated noisy Lagrangian tracers transported by the flows. Understanding the evolution of the joint PDF is potentially important for determining the optimal number and the best locations of releasing the tracers \citep{chen2016model}.

\section*{Acknowledgement}
The research of A.J.M. is partially supported by the Office of Naval Research Grant ONR MURI N00014-16-1-2161 and the Center for Prototype Climate Modeling (CPCM) at New York University Abu Dhabi Research Institute. N.C. is supported as a postdoctoral fellow through A.J.M's ONR MURI Grant. The authors thank Yoonsang Lee and Xin Tong for useful discussion.
\appendix
\section{A General Framework of Conditional Gaussian Systems with Energy-Conserving Nonlinear Interactions}\label{Appendix:A}
Recall the general form of turbulent dynamical system with energy-conserving quadratic nonlinear interactions in \eqref{EnergyConserveModel}:
\begin{equation}\label{EnergyConserveModel_Appendix}
\begin{split}
  d\mathbf{u} &= \big[ (\mathbf{L}+\mathbf{D})\mathbf{u} + \mathbf{B}(\mathbf{u},\mathbf{u}) + \mathbf{F}(t) \big] dt + \boldsymbol{\Sigma}(t,\mathbf{u})d\mathbf{W}(t),\\
  &\mbox{with}\qquad\mathbf{u}\cdot\mathbf{B}(\mathbf{u},\mathbf{u}) = 0.
\end{split}
\end{equation}
To find the class of models that belong to the conditional Gaussian framework \eqref{Conditional_Gaussian_System}, we rewrite the equation \eqref{EnergyConserveModel_Appendix} in the following way
\begin{equation}\label{EnergyConserveModel2_Appendix}
\begin{split}
  d\mathbf{u}_\mathbf{I} &= \big(\mathbf{L}_\mathbf{I,1}\mathbf{u}_\mathbf{I} +\mathbf{L}_\mathbf{I,2}\mathbf{u}_\mathbf{II}  + \mathbf{B}_\mathbf{I}(\mathbf{u},\mathbf{u}) + \mathbf{F}_\mathbf{I} \big) dt + \boldsymbol{\Sigma}_\mathbf{I}(\mathbf{u}_\mathbf{I})d\mathbf{W}_\mathbf{I},\\
  d\mathbf{u}_\mathbf{II} &= \big(\mathbf{L}_\mathbf{II,1}\mathbf{u}_\mathbf{I} +\mathbf{L}_\mathbf{II,2}\mathbf{u}_\mathbf{II}  + \mathbf{B}_\mathbf{II}(\mathbf{u},\mathbf{u}) + \mathbf{F}_\mathbf{II} \big) dt + \boldsymbol{\Sigma}_\mathbf{II}(\mathbf{u}_\mathbf{I})d\mathbf{W}_\mathbf{II},
\end{split}
\end{equation}
where the explicit dependence of the coefficients on time $t$ has been omitted for  notation simplicity. In \eqref{EnergyConserveModel2_Appendix}, $\mathbf{L}_\mathbf{I,1}\mathbf{u}_\mathbf{I}$, $\mathbf{L}_\mathbf{I,2}\mathbf{u}_\mathbf{II}$, $\mathbf{L}_\mathbf{II,1}\mathbf{u}_\mathbf{I}$ and $\mathbf{L}_\mathbf{II,2}\mathbf{u}_\mathbf{II}$ correspond to the the linear term $\mathbf{L}+\mathbf{D}$ in \eqref{EnergyConserveModel_Appendix}
while $\mathbf{B}_\mathbf{I}(\mathbf{u},\mathbf{u})$ and $\mathbf{B}_{\mathbf{II}}(\mathbf{u},\mathbf{u})$ represent the nonlinear terms in the processes associated with the observed variables \eqref{Conditional_Gaussian_System1} and unobserved variables \eqref{Conditional_Gaussian_System2}, respectively. Since the conditional Gaussian systems do not allow quadratic nonlinear interactions between $\mathbf{u}_\mathbf{II}$ and itself, both $\mathbf{B}_\mathbf{I}(\mathbf{u},\mathbf{u})$ and $\mathbf{B}_{\mathbf{II}}(\mathbf{u},\mathbf{u})$ can be written down in the following forms
\begin{equation}\label{Decompose_B}
\begin{split}
  \mathbf{B}_\mathbf{I}(\mathbf{u},\mathbf{u}) &= \mathbf{B}_\mathbf{I,1}(\mathbf{u}_\mathbf{I},\mathbf{u}_\mathbf{I}) + \mathbf{B}_\mathbf{I,2}(\mathbf{u}_\mathbf{I},\mathbf{u}_\mathbf{II})\\
  \mathbf{B}_{\mathbf{II}}(\mathbf{u},\mathbf{u}) &= \mathbf{B}_\mathbf{II,1}(\mathbf{u}_\mathbf{I},\mathbf{u}_\mathbf{I}) + \mathbf{B}_\mathbf{II,2}(\mathbf{u}_\mathbf{I},\mathbf{u}_\mathbf{II})
\end{split}
\end{equation}
where $\mathbf{B}_\mathbf{\cdot,1}(\mathbf{u}_\mathbf{I},\mathbf{u}_\mathbf{I})$ stands for the quadratic terms involving only $\mathbf{u}_\mathbf{I}$ and $\mathbf{B}_\mathbf{\cdot,2}(\mathbf{u}_\mathbf{I},\mathbf{u}_\mathbf{II})$ represents the quadratic interactions between $\mathbf{u}_\mathbf{I}$ and $\mathbf{u}_\mathbf{II}$.  Given the nonlinear terms in  \eqref{Decompose_B}, the energy-conserving quadratic nonlinearity in \eqref{EnergyConserveModel_Appendix} implies
\begin{equation}\label{EnergyConserving_B}
  \mathbf{u}_\mathbf{I}\cdot\Big(\mathbf{B}_\mathbf{I,1}(\mathbf{u}_\mathbf{I},\mathbf{u}_\mathbf{I}) + \mathbf{B}_\mathbf{I,2}(\mathbf{u}_\mathbf{I},\mathbf{u}_\mathbf{II})\Big) + \mathbf{u}_\mathbf{II}\cdot\Big(\mathbf{B}_\mathbf{II,1}(\mathbf{u}_\mathbf{I},\mathbf{u}_\mathbf{I}) + \mathbf{B}_\mathbf{II,2}(\mathbf{u}_\mathbf{I},\mathbf{u}_\mathbf{II})\Big) = 0.
\end{equation}
Inserting \eqref{Decompose_B} into \eqref{EnergyConserveModel2_Appendix} yields the conditional Gaussian systems with energy-conserving quadratic nonlinear interactions,
\begin{subequations}\label{Conditional_Gaussian_System_Quadratic}
\begin{align}
    d\mathbf{u}_{\mathbf{I}} &= \Big(\mathbf{B}_\mathbf{I,1}(\mathbf{u}_\mathbf{I},\mathbf{u}_\mathbf{I}) + \mathbf{B}_\mathbf{I,2}(\mathbf{u}_\mathbf{I},\mathbf{u}_\mathbf{II}) + \mathbf{L}_\mathbf{I,1}\mathbf{u}_{\mathbf{I}} + \mathbf{L}_\mathbf{I,2}\mathbf{u}_{\mathbf{II}} + \mathbf{F}_{\mathbf{I}} \Big)dt + \boldsymbol{\Sigma}_{\mathbf{I}}(\mathbf{u}_{\mathbf{I}})d\mathbf{W}_{\mathbf{I}},\label{Conditional_Gaussian_System_Quadratic1}\\
    d\mathbf{u}_{\mathbf{II}} &= \Big(\mathbf{B}_\mathbf{II,1}(\mathbf{u}_\mathbf{I},\mathbf{u}_\mathbf{I}) + \mathbf{B}_\mathbf{II,2}(\mathbf{u}_\mathbf{I},\mathbf{u}_\mathbf{II}) + \mathbf{L}_\mathbf{II,1}\mathbf{u}_{\mathbf{I}} + \mathbf{L}_\mathbf{II,2}\mathbf{u}_{\mathbf{II}} + \mathbf{F}_{\mathbf{II}}\Big)dt + \boldsymbol{\Sigma}_{\mathbf{II}}(\mathbf{u}_{\mathbf{I}})d\mathbf{W}_{\mathbf{II}}, \label{Conditional_Gaussian_System_Quadratic2}
\end{align}
\end{subequations}

Now we explore the detailed forms of the energy-conserving nonlinear terms in \eqref{Conditional_Gaussian_System_Quadratic}.

We start with $\mathbf{B}_\mathbf{II,2}(\mathbf{u}_\mathbf{I},\mathbf{u}_\mathbf{II})$, which can be written as
\begin{equation}\label{Component_B22}
  \mathbf{B}_\mathbf{II,2}(\mathbf{u}_\mathbf{I},\mathbf{u}_\mathbf{II}) = \mathbf{S}_\mathbf{II}(\mathbf{u}_\mathbf{I})\mathbf{u}_\mathbf{II},\qquad \mbox{with}\qquad \mathbf{S}_\mathbf{II}(\mathbf{u}_\mathbf{I}) = \sum_{j=1}^{N_I}S_{\mathbf{II},j} u_{\mathbf{I},j},
\end{equation}
where each $S_{\mathbf{II},j}$ is a skew-symmetric matrix with $S^T_{\mathbf{II},j} = -S_{\mathbf{II},j}$ and $u_{\mathbf{I},j}$ is the $j$-th entry of $\mathbf{u}_{\mathbf{I}}$. The energy-conserving property is easily seen by multiplying $\mathbf{u}_{\mathbf{II}}$ to $\mathbf{B}_\mathbf{II,2}(\mathbf{u}_\mathbf{I},\mathbf{u}_\mathbf{II})$ in \eqref{Component_B22},
\begin{equation*}
  \mathbf{u}_{\mathbf{II}}\cdot\mathbf{B}_\mathbf{II,2}(\mathbf{u}_\mathbf{I},\mathbf{u}_\mathbf{II}) = \mathbf{u}_{\mathbf{II}}\cdot\mathbf{S}(\mathbf{u}_\mathbf{I})\cdot\mathbf{u}_\mathbf{II} = \sum_{j=1}^{N_I} u_{\mathbf{I},j}\cdot\Big(\mathbf{u}_{\mathbf{II}}\cdot S_j\cdot\mathbf{u}_{\mathbf{II}}\Big)=0,
\end{equation*}
due to the skew-symmetric property of $S_j$. In fact, $\mathbf{B}_\mathbf{II,2}(\mathbf{u}_\mathbf{I},\mathbf{u}_\mathbf{II})$ usually represents the internal oscillation with non-constant oscillation frequency that depends on $\mathbf{u}_\mathbf{I}$.

Next, $\mathbf{B}_\mathbf{I,2}(\mathbf{u}_\mathbf{I},\mathbf{u}_\mathbf{II})$ contains three components,
\begin{equation}\label{BI2_Appendix}
  \mathbf{B}_\mathbf{I,2}(\mathbf{u}_\mathbf{I},\mathbf{u}_\mathbf{II}) = \mathbf{B}^{\mathbf{1}}_\mathbf{I,2}(\mathbf{u}_\mathbf{I},\mathbf{u}_\mathbf{II}) + \mathbf{B}^{\mathbf{2}}_\mathbf{I,2}(\mathbf{u}_\mathbf{I},\mathbf{u}_\mathbf{II}) + \mathbf{B}^{\mathbf{3}}_\mathbf{I,2}(\mathbf{u}_\mathbf{I},\mathbf{u}_\mathbf{II}).
\end{equation}
One of the components in \eqref{BI2_Appendix}, say $\mathbf{B}^\mathbf{1}_\mathbf{I,2}(\mathbf{u}_\mathbf{I},\mathbf{u}_\mathbf{II})$, has its own  energy conservation, i.e.,
\begin{equation*}
  \mathbf{u}_\mathbf{I}\cdot \mathbf{B}^\mathbf{1}_\mathbf{I,2}(\mathbf{u}_\mathbf{I},\mathbf{u}_\mathbf{II}) = 0.
\end{equation*}
Here, $\mathbf{B}^\mathbf{1}_\mathbf{I,2}(\mathbf{u}_\mathbf{I},\mathbf{u}_\mathbf{II}) = \mathbf{S}_{\mathbf{I}}(\mathbf{u}_{\mathbf{I}})\mathbf{u}_{\mathbf{II}}$ and therefore
\begin{equation}\label{Form_S_II_Energy}
  \mathbf{u}_\mathbf{I}\cdot \mathbf{S}_{\mathbf{I}}(\mathbf{u}_{\mathbf{I}})\mathbf{u}_{\mathbf{II}} = 0,
\end{equation}
where each column of $\mathbf{S}_{\mathbf{I}}(\mathbf{u}_{\mathbf{I}})$ is given by
\begin{equation}\label{Form_S_II}
  \mathbf{S}_{\mathbf{I},j}(\mathbf{u}_{\mathbf{I}}) = S_{\mathbf{I},j}\mathbf{u}_{\mathbf{I}},
\end{equation}
with $S_{\mathbf{I},j}$ being a skew-symmetric matrix. Thus, with \eqref{Form_S_II} in hand, \eqref{Form_S_II_Energy} becomes
\begin{equation*}
  \sum_{j=1}^{N_\mathbf{II}}\Big(\mathbf{u}_\mathbf{I}\cdot S_{\mathbf{I},j}\cdot\mathbf{u}_{\mathbf{I}}\Big)\mathbf{u}_{\mathbf{II},j} = 0,
\end{equation*}
where $\mathbf{u}_{\mathbf{II},j}$ is the $j$-th entry of $\mathbf{u}_{\mathbf{II}}$.

The other two components of $\mathbf{B}_\mathbf{I,2}(\mathbf{u}_\mathbf{I},\mathbf{u}_\mathbf{II})$ in \eqref{Conditional_Gaussian_System_Quadratic1} involve the interactions with $\mathbf{B}_\mathbf{II,1}(\mathbf{u}_\mathbf{I},\mathbf{u}_\mathbf{I}) = \mathbf{B}^{\mathbf{2}}_\mathbf{II,1}(\mathbf{u}_\mathbf{I},\mathbf{u}_\mathbf{I}) + \mathbf{B}^{\mathbf{3}}_\mathbf{II,1}(\mathbf{u}_\mathbf{I},\mathbf{u}_\mathbf{I})$ in \eqref{Conditional_Gaussian_System_Quadratic2}. On one hand, the energy-conserving property in the following two terms is obvious,
\begin{subequations}\label{Dyad_Interaction}
\begin{align}
  \mathbf{B}^{\mathbf{2}}_\mathbf{I,2}(\mathbf{u}_\mathbf{I},\mathbf{u}_\mathbf{II}) &= \sum_{j=1}^{N_I}\Gamma_j \mathbf{u}_{\mathbf{I},j}\mathbf{u}_\mathbf{II},\label{Dyad_Interaction1}\\
  \mathbf{B}^{\mathbf{2}}_\mathbf{II,1}(\mathbf{u}_\mathbf{I},\mathbf{u}_\mathbf{I}) & = -\sum_{j=1}^{N_I} \Gamma_j^T \mathbf{u}_{\mathbf{I}}^2,\label{Dyad_Interaction2}
\end{align}
\end{subequations}
where each $\Gamma_j$ is a $N_\mathbf{I}\times N_\mathbf{II}$ matrix, $\mathbf{u}_{\mathbf{I},j}$ is the $j$-th entry of $\mathbf{u}_{\mathbf{I}}$ and $\mathbf{u}_{\mathbf{I}}^2$ is a vector of size $N_\mathbf{I}\times1$ with the $j$-th entry being $\mathbf{u}_{\mathbf{I},j}^2$. On the other hand, the remaining two terms $\mathbf{B}^{\mathbf{3}}_\mathbf{I,2}(\mathbf{u}_\mathbf{I},\mathbf{u}_\mathbf{II})$ and $\mathbf{B}^{\mathbf{3}}_\mathbf{II,1}(\mathbf{u}_\mathbf{I},\mathbf{u}_\mathbf{I})$ are similar to those in \eqref{Dyad_Interaction} but deal with the cross-interactions between different components of $\mathbf{u}_{\mathbf{I}}$ such as replacing $\mathbf{u}_{\mathbf{I}}^2$ by $\mathbf{u}_{\mathbf{I},j_1}\mathbf{u}_{\mathbf{I},j_2}$ in \eqref{Dyad_Interaction2}. To this end, we define the following
\begin{equation}\label{Dyad_Interaction2}
  \mathbf{G}(\mathbf{u}_{\mathbf{I}}) = \sum_{j=1}^{N_I} G_j \mathbf{u}_{\mathbf{I},j},
\end{equation}
which satisfies
\begin{equation}\label{Dyad_Interaction2_condition}
  \mathbf{u}_\mathbf{I}\cdot\mathbf{G}(\mathbf{u}_{\mathbf{I}})\mathbf{u}_\mathbf{II} - \mathbf{u}_\mathbf{II} \cdot\mathbf{G}^T(\mathbf{u}_{\mathbf{I}})\mathbf{u}_\mathbf{I} = 0
\end{equation}
In fact, \eqref{Dyad_Interaction}--\eqref{Dyad_Interaction2_condition} are important for generating the intermittent instability, where $\mathbf{u}_{\mathbf{II}}$ plays the role of both damping and anti-damping for the dynamics of $\mathbf{u}_{\mathbf{I}}$.

Finally, $\mathbf{B}_\mathbf{I,1}(\mathbf{u}_\mathbf{I},\mathbf{u}_\mathbf{I})$ involves any iterations between $\mathbf{u}_{\mathbf{I}}$ and itself that satisfies
\begin{equation}\label{BI1}
  \mathbf{u}_{\mathbf{I}}\cdot\mathbf{B}_\mathbf{I,1}(\mathbf{u}_\mathbf{I},\mathbf{u}_\mathbf{I}) = 0.
\end{equation}

Therefore, with \eqref{Component_B22}--\eqref{BI1} in hand, the conditional Gaussian system \eqref{Conditional_Gaussian_System_Quadratic} has the following form,
\begin{subequations}\label{Conditional_Gaussian_System_Details}
\begin{align}
    d\mathbf{u}_{\mathbf{I}} &= \Big(\mathbf{B}_\mathbf{I,1}(\mathbf{u}_\mathbf{I},\mathbf{u}_\mathbf{I}) + \sum_{j=1}^{N_\mathbf{I}}\Gamma_j u_{\mathbf{I},j}\mathbf{u}_\mathbf{II} + \mathbf{S}_{\mathbf{I}}(\mathbf{u}_{\mathbf{I}})\mathbf{u}_{\mathbf{II}} + \mathbf{G}(\mathbf{u}_{\mathbf{I}})\mathbf{u}_{\mathbf{II}} +\mathbf{L}_\mathbf{I,1}\mathbf{u}_{\mathbf{I}} \notag\\
     &\qquad\qquad + \mathbf{L}_\mathbf{I,2}\mathbf{u}_{\mathbf{II}} + \mathbf{F}_{\mathbf{I}} \Big)dt + \boldsymbol{\Sigma}_{\mathbf{I}}(\mathbf{u}_{\mathbf{I}})d\mathbf{W}_{\mathbf{I}},\label{Conditional_Gaussian_System_Details1}\\
    d\mathbf{u}_{\mathbf{II}} &= \Big(  \mathbf{S}_\mathbf{II}(\mathbf{u}_\mathbf{I})\mathbf{u}_\mathbf{II} - \sum_{j=1}^{N_\mathbf{I}} \Gamma_j^T \mathbf{u}_{\mathbf{I}}^2 + \mathbf{L}_\mathbf{II,1}\mathbf{u}_{\mathbf{I}} + \mathbf{L}_\mathbf{II,2}\mathbf{u}_{\mathbf{II}} - \mathbf{G}^T(\mathbf{u}_{\mathbf{I}})\mathbf{u}_{\mathbf{I}}\notag\\
      &\qquad\qquad\qquad\qquad + \mathbf{F}_{\mathbf{II}}\Big)dt + \boldsymbol{\Sigma}_{\mathbf{II}}(\mathbf{u}_{\mathbf{I}})d\mathbf{W}_{\mathbf{II}}. \label{Conditional_Gaussian_System_Details2}
\end{align}
\end{subequations}

In addition to the four examples introduced in Section \ref{Sec:Systems}, another representative example is the advective two-layer Lorenz-96 model \citep{lee2017multiscale}
\begin{equation}\label{model}
	\centering
	\begin{split}
	\frac{du_i}{dt}&=u_{i-1}(u_{i+1}-u_{i-2})+\lambda\sum_{j=1}^{J}v_{i,j}-d_1u_i+F+\sigma_{u_i}\dot{W}_{u_i}, \quad i=1,2,...,I\\
	\frac{dv_{i,j}}{dt}&=\frac{a_Lu_i+a_Sv_{i,j+1}}{\epsilon}(v_{i,j-1}-v_{i,j+2})-\lambda u_i-d_2v_{i,j}, \quad j=1,2,...,J\\
	\end{split}
\end{equation}
where $u_i$ is periodic in $i$ and $v_{ij}$ is periodic in both $i$ and $j$. This model is developed as a test model for multiscale data assimilation methods. As a special case of this model, the model with $a_S=0$, which is a slow-fast system, fits into the conditional Gaussian model framework \eqref{Conditional_Gaussian_System} with $\mathbf{u}_\mathbf{I}=\{u_i\}$ and $\mathbf{u}_\mathbf{II}=\{v_{i,j}\}$. In \citep{lee2017multiscale}, it is shown that the model with appropriate parameters shows non-Gaussian fat-tails in both the observed and hidden variables. As the dimension $I$ and $J$ can be manipulated, this model is a good candidate for the uncertainty quantification and recovering PDFs of high-dimensional systems using the conditional Gaussianity.

\section{Kernel Density Estimation with a Solve-The-Equation Bandwidth}\label{Appendix:B}
Here we summarize the basic idea of the kernel density estimation method that is adopted in this article to solve the distribution $p(\mathbf{u}_{\mathbf{I}})$. We first discuss the idea based on 1D case. Then we describe the multi-dimensional case.

Assume we have $L$ observational data points $u^i, i,\ldots, L$ at a fixed time. The approximation of the unknown 1D PDF $p(u)$ is given by the kernel density estimator
\begin{equation}\label{Kernel_Method}
  \hat{p}_h(u) = \frac{1}{L}\sum_{i=1}^L K_h(u-u^i) = \frac{1}{Lh}\sum_{i=1}^L K\left(\frac{u-u^i}{h}\right),
\end{equation}
where $K(\cdot)$ is the kernel with $K>0$ and $\int K dx= 1$ and $h$ is the so-called the bandwidth that is a crucial parameter for the kernel density estimation. The kernel $K(\cdot)$ has different choices,  and a Gaussian kernel is adopted in the main text.

One of the most commonly used criteria for selecting $h$ is to minimize the mean integrated squared error (MISE):
\begin{equation*}
  \mbox{MISE}(h) = E\left[\int(\hat{p}_h(u)-p(u))^2dx\right].
\end{equation*}
Under the weak assumptions on $p$ and $K$ \cite{rosenblatt1956remarks, parzen1962estimation}, MISE$(h)$ = AMISE$(h)+o(1/(Lh)+h^4)$, where AMISE is the asymptotic MISE and it is given by
\begin{equation*}
  \mbox{AMISE}(h) = \frac{R(K)}{Lh} + \frac{1}{4}m_2(K)^2h^4R(p^{\prime\prime}),
\end{equation*}
with $R(K) = \int K(u)^2 du$, $m_2(K) = \int u^2 K(u)du$ and $p^{\prime\prime}$ being the second derivative of $p$. The minimum of the AMISE is the solution to the following differential equation
\begin{equation*}
  \frac{\partial}{\partial h}\mbox{AMISE}(h) = -\frac{R(K)}{Lh^2} + m_2(K)^2h^3R(p^{\prime\prime}) = 0,
\end{equation*}
the solution of which is given by
\begin{equation}\label{AMISE}
  h_{\mbox{\tiny AMISE}} = \frac{R(K)^{1/5}}{m_2(K)^{2/5}\,R(p^{\prime\prime})^{1/5}\,L^{1/5}}.
\end{equation}
Unfortunately, there is no explicit solution for $h_{\mbox{\tiny AMISE}}$ in \eqref{AMISE} that applies for a general density function $p(u)$ since \eqref{AMISE} involves the  unknown density function $p$ and its second derivative $p^{\prime\prime}$. Under the assumption that the true density is Gaussian, the rule-of-thumb bandwidth estimator can be adopted for solving the optimal bandwidth with explicit expressions. However, the typical PDFs in turbulent dynamical systems are far from Gaussian and the rule-of-thumb bandwidth estimators fail to capture the non-Gaussian features. A practical approximation is to use the ``solve-the-equation plug-in principle'', namely using $\hat{p}^{\prime\prime}$ to replace $p^{\prime\prime}$ in \eqref{AMISE} to solve $h_{\mbox{\tiny AMISE}}$  \cite{botev2010kernel, raykar2006fast, jones1996brief, alexandre2008solve}. The one we adopted in the main text is from \cite{botev2010kernel}, which is free from the arbitrary normal reference rules and its skill has been shown in recovering the highly non-Gaussian PDF.

For multi-dimensional case, the kernel density estimation is defined as
\begin{equation*}
  \hat{p}_{\mathbf{H}}(\mathbf{u}) = \frac{1}{L}\sum_{i=1}^L K_{\mathbf{H}}(\mathbf{u}-\mathbf{u}^i),
\end{equation*}
where $\mathbf{u}=(u_1,\ldots,u_d)^T$ and $\mathbf{H}$ is the bandwidth $d\times d$ matrix that is symmetric and positive definite. The kernel function is a multivariate density. Again, as in the main text, we use a multivariate normal kernel density,
\begin{equation*}
  K_{\mathbf{H}}(\mathbf{u}) = (2\pi)^{-d/2}|\mathbf{H}|^{-1/2}e^{-\frac{1}{2}\mathbf{u}^T\mathbf{H}^{-1}\mathbf{u}}.
\end{equation*}
There are different ways of forming the kernel matrix $\mathbf{H}$. For example, $\mathbf{H}$ can be assumed to be a full matrix, or simplified as a diagonal matrix or even a multiplier of a unit matrix. Here, we adopt a diagonal matrix for $\mathbf{H}$. This greatly reduces the computational costs while remains the results with reasonable accuracy. Nevertheless, the optimal bandwidth in the $(i,i)$-th diagonal entry of $\mathbf{H}$ does not equal to the optimal bandwidth of the corresponding 1D problem, since the minimization of the MISE in the target function here involves the multi-dimensional density.

\section{Convergence of Gaussian Mixture Distribution with and without Off-Diagonal Block Components in Each Component}\label{Appendix:C}
In this Appendix, we show that the Gaussian mixture with each component given by \eqref{MeanCov} in the efficient statistically accurate algorithm (Proposition \ref{Prop:Both}) that contains a block diagonal covariance matrix will converge to the same distribution with a Gaussian mixture that the off-diagonal block components are nonzero.

To this end, consider the two distributions as follows:
\begin{equation}\label{pxy}
\begin{split}
  p(\mathbf{u}_{\mathbf{I}},\mathbf{u}_{\mathbf{II}}) &= \lim_{L\to\infty}\frac{1}{L}\sum_{i=1}^L p_i(\mathbf{u}_{\mathbf{I}},\mathbf{u}_{\mathbf{II}}),\\
  \tilde{p}(\mathbf{u}_{\mathbf{I}},\mathbf{u}_{\mathbf{II}}) &= \lim_{L\to\infty}\frac{1}{L}\sum_{i=1}^L \tilde{p}_i(\mathbf{u}_{\mathbf{I}},\mathbf{u}_{\mathbf{II}}),
\end{split}
\end{equation}
where for each $i = 1,\ldots L$,
\begin{equation}\label{pj}
  p_i\sim \mathcal{N}(\boldsymbol{\mu}_i,\Sigma_i),\qquad\mbox{and}\qquad
  \tilde{p}_i\sim \mathcal{N}(\boldsymbol{\mu}_i,\tilde{\Sigma}_i).
\end{equation}
Here, $\Sigma_i$ is a full matrix while $\tilde{\Sigma}_i$ is a block diagonal matrix as in \eqref{MeanCov},
\begin{equation}\label{Sigmaj}
  \Sigma_i = \left(
               \begin{array}{cc}
                 \Sigma_{i,11} & \Sigma_{i,12} \\
                 \Sigma_{i,21} & \Sigma_{i,22} \\
               \end{array}
             \right),\qquad\mbox{and}\qquad
  \tilde{\Sigma}_i = \left(
               \begin{array}{cc}
                 \tilde{\Sigma}_{i,11} & 0\\
                 0 & \Sigma_{i,22} \\
               \end{array}\right)
\end{equation}
The difference between $\Sigma_{j,11}$ and $\tilde{\Sigma}_{j,11}$ is allowed since the bandwidth in the kernel estimation can be different. But it is required that the decay rates of the elements in $\Sigma_{j,11}$ and $\tilde{\Sigma}_{j,11}$ as a function of $L$ have the same order, i.e., both being $L^{-\delta}$ with $\delta>0$. The other part $\Sigma_{i,22}$ is from the conditional Gaussian posterior distribution and is assumed to be the same in $\Sigma_i$ and $\tilde{\Sigma}_i$.

Now we make use of the characteristic functions to show that the error between  $p(\mathbf{u}_\mathbf{I},\mathbf{u}_\mathbf{II})$ and $\tilde{p}(\mathbf{u}_\mathbf{I},\mathbf{u}_\mathbf{II})$ goes to zero as the number of observational trajectories $L$ goes to infinity. Since the characteristic function and the PDF have one-to-one correspondence, it is sufficient to show that the error in the associated characteristic functions goes to zero.
The definition of a $k$-dimension vector $\mathbf{z}$ is given by
\begin{equation*}
  \psi_\mathbf{z}(\mathbf{w}) = \mathbb{E}\big[\exp(\imath\mathbf{w}^T\mathbf{z})\big],
\end{equation*}
where $\mathbf{w}\in \mathbb{R}^k$ and $\imath$ is the imaginary unit. Particularly, if $\mathbf{z}$ is a multi-dimensional Gaussian variable $\mathcal{N}(\boldsymbol{\mu},\Sigma)$, then its characteristic function is given by
\begin{equation}\label{Charac_Gaussian}
  \psi_\mathbf{z}(\mathbf{w}) = e^{\imath\mathbf{w}^T\boldsymbol{\mu} - \frac{1}{2}\mathbf{w}^T\Sigma\mathbf{w}}.
\end{equation}

\begin{proposition}
Denote the characteristic functions of $p(\mathbf{u}_{\mathbf{I}},\mathbf{u}_{\mathbf{II}})$ and $\tilde{p}(\mathbf{u}_{\mathbf{I}},\mathbf{u}_{\mathbf{II}})$ by $\psi(\mathbf{w})$ and $\tilde\psi(\mathbf{w})$, respectively. With a sufficiently large $L$, the following result holds:
\begin{equation}\label{Charac_error}
  |\psi(\mathbf{w})-\tilde\psi(\mathbf{w})| \leq DL^{-\delta},
\end{equation}
where $D$ is a constant and $L^{-\delta}$ is the decay rate of the bandwidth as a function of $L$ in the kernel density estimation.
\end{proposition}
\begin{proof}
The characteristic functions of $p(\mathbf{u}_{\mathbf{I}},\mathbf{u}_{\mathbf{II}})$ and $\tilde{p}(\mathbf{u}_{\mathbf{I}},\mathbf{u}_{\mathbf{II}})$ are given respectively by
\begin{equation*}
  \psi(\mathbf{w}) = \lim_{L\to\infty}\frac{1}{L}\sum_{i=1}^L\psi_i(\mathbf{w})\qquad\mbox{and}\qquad
  \tilde\psi(\mathbf{w})= \lim_{L\to\infty}\frac{1}{L}\sum_{i=1}^L\tilde\psi_i(\mathbf{w}).
\end{equation*}
The error between $\psi(\mathbf{w})$ and $\tilde\psi(\mathbf{w})$ yields
\begin{equation}\label{psidiff}
\begin{split}
  |\psi(\mathbf{w})-\tilde\psi(\mathbf{w})| &= \lim_{L\to\infty}\left|\frac{1}{L}\sum_{i=1}^L\left(\psi_i(\mathbf{w})-\tilde\psi_i(\mathbf{w})\right)\right|\\
  &\leq\lim_{L\to\infty}\frac{1}{L}\sum_{i=1}^L\left|\psi_i(\mathbf{w})-\tilde\psi_i(\mathbf{w})\right|.
\end{split}
\end{equation}
Below, we focus on the error in each Gaussian component $\left|\psi_i(\mathbf{w})-\tilde\psi_i(\mathbf{w})\right|$. For the simplicity of notation, we omit the subscript $i$ in the mean and covariance, namely we use the notations
\begin{equation*}
  \boldsymbol\mu :=\boldsymbol\mu_i,\qquad \Sigma:=\Sigma_i,\qquad\mbox{and}\qquad \tilde{\Sigma}:=\tilde\Sigma_i.
\end{equation*}
In light of the explicit expression of the characteristic function associated with the multivariate Gaussian in \eqref{Charac_Gaussian}, we have
\begin{align}
  \left|\psi_i(\mathbf{w})-\tilde\psi_i(\mathbf{w})\right| &=\left|e^{\imath\mathbf{w}^T\boldsymbol\mu-\frac{1}{2}\mathbf{w}^T\Sigma\mathbf{w}} - e^{\imath\mathbf{w}^T\boldsymbol\mu-\frac{1}{2}\mathbf{w}^T\tilde\Sigma\mathbf{w}}\right| \notag\\
  & = \left|e^{\imath\mathbf{w}^T\boldsymbol\mu}\right|\left|e^{-\frac{1}{2}\mathbf{w}^T\Sigma\mathbf{w}} - e^{-\frac{1}{2}\mathbf{w}^T\tilde\Sigma\mathbf{w}}\right| \notag\\
  & = \left|e^{-\frac{1}{2}\mathbf{w}^T\Sigma\mathbf{w}} - e^{-\frac{1}{2}\mathbf{w}^T\tilde\Sigma\mathbf{w}}\right| \label{diff1}\\
  & = \left|e^{-\frac{1}{2}\mathbf{w}^T\Sigma\mathbf{w}}\left( e^{-\frac{1}{2}\mathbf{w}^T(\tilde\Sigma-\Sigma)\mathbf{w}}-1\right)\right|\notag\\
  & \leq \left| e^{-\frac{1}{2}\mathbf{w}^T(\tilde\Sigma-\Sigma)\mathbf{w}}-1\right|.\label{diff2}
\end{align}
Since both the covariance matrices $\Sigma$ and $\tilde\Sigma$ are positive definite, according to \eqref{diff1} there exists a large positive number $M$ such that when $|\mathbf{w}|>M$,
\begin{equation*}
  \left|\psi_i(\mathbf{w})-\tilde\psi_i(\mathbf{w})\right| \to 0.
\end{equation*}
On the other hand, in light of \eqref{Sigmaj}, we have
\begin{equation*}
  \tilde{\Sigma}-\Sigma = \left(
                            \begin{array}{cc}
                              \Sigma_{11}-\tilde{\Sigma}_{11} & \Sigma_{12} \\
                              \Sigma_{21} & 0 \\
                            \end{array}
                          \right),
\end{equation*}
where each component has the following form
\begin{equation*}
\begin{split}
    (\Sigma_{11}-\tilde{\Sigma}_{11})_{k_1k_2} &= \sigma_{k_1,\mathbf{x}}\sigma_{k_2,\mathbf{x}} - \tilde\sigma_{k_1,\mathbf{x}}\tilde\sigma_{k_2,\mathbf{x}},\\
    (\Sigma_{12})_{kl} &= \rho_{kl}\sigma_{k,\mathbf{x}}\sigma_{l,\mathbf{y}}.
\end{split}
\end{equation*}
All the entries of $\Sigma_{11}-\tilde{\Sigma}_{11}$ are bounded by $\tilde{c}L^{-2\delta}$ and those of  $\Sigma_{12}$ and $\Sigma_{21}$ are bounded by $\tilde{c}L^{-\delta}$. For a fixed $\mathbf{w}$ with $|\mathbf{w}|\leq M$, there exists an $L$ such that
\begin{equation}\label{diff3}
  \left|\mathbf{w}^T(\tilde\Sigma-\Sigma)\mathbf{w}\right|\leq \tilde{c}_0L^{-\delta},
\end{equation}
where $|\cdot|$ is the absolute value not the determinant.
Adopting the Taylor expansion of \eqref{diff2} and making use of \eqref{diff3} yields
\begin{equation}\label{diff4}
  \left|\psi_i(\mathbf{w})-\tilde\psi_i(\mathbf{w})\right| \leq c_1L^{-\delta}.
\end{equation}
With \eqref{diff4} in hand, it is straightforward to arrive at the conclusion with respect to \eqref{psidiff} that
\begin{equation}\label{diff5}
  \left|\psi(\mathbf{w})-\tilde\psi(\mathbf{w})\right| \leq DL^{-\delta}
\end{equation}
\end{proof}

It can be shown that the bandwidth selector $\mathbf{H}$ has $\mathbf{H}=O(n^{-2/(N_\mathbf{I}+4)})$ elementwise \cite{wand1995kernel}. If we denote $\rho_{ij}\sigma_{i}\sigma_{j}$ as the $(i,j)$-component of $\mathbf{H}$, then $\sigma_{i}=O(n^{-1/(N_\mathbf{I}+4)})$ for all $i=1,\ldots,d_1$, which implies $\delta = \frac{1}{N_\mathbf{I}+4}$. Therefore, \eqref{diff5} becomes
\begin{equation}\label{diff6}
  \left|\psi(\mathbf{w})-\tilde\psi(\mathbf{w})\right| \leq Dn^{-\frac{1}{N_\mathbf{I}+4}}.
\end{equation}

\end{document}